\documentclass[11pt,draftcls,onecolumn]{IEEEtran}

\ifCLASSINFOpdf
\else
   \usepackage[dvips]{graphicx}
\fi
\usepackage{url}
\usepackage{graphicx}
\usepackage{amsmath}
\usepackage{comment}
\usepackage[utf8]{inputenc}
\usepackage[english]{babel}
\usepackage{color}
\usepackage{framed}
\usepackage{cite}
\usepackage{amsmath,amssymb,amsfonts}
\usepackage{graphicx}
\usepackage{textcomp}
\usepackage{cite}
\usepackage{bbold}
\usepackage{ifpdf}
\usepackage{times}
\usepackage{layout}
\usepackage{float}
\usepackage{afterpage}
\usepackage{amsmath}
\usepackage{amstext}
\usepackage{amssymb,bm}
\usepackage{latexsym}
\usepackage{color}
\usepackage{flexisym}
\usepackage{dblfloatfix}    
\usepackage{kantlipsum}
\usepackage{subcaption}
\usepackage{graphicx}
\usepackage{amsmath}
\usepackage{amsthm}
\usepackage{upgreek}
\usepackage{graphicx}
\usepackage{caption}
\usepackage{subcaption}
\usepackage{booktabs}
\usepackage{multicol}
\usepackage{lipsum}

\usepackage{enumitem}
\usepackage{amsmath}
\usepackage{algpseudocode}
\usepackage{amsthm}
\usepackage{dsfont}
\usepackage{algorithm}
\usepackage{thmtools}
\usepackage{amssymb}
\usepackage[dvipsnames]{xcolor}

\def\bb{\boldsymbol{b}}

\def\br{\boldsymbol{r}}

\def\bA{\boldsymbol{A}}

\def\bR{\boldsymbol{R}}

\def\bGamma{\boldsymbol{\Gamma}}

\def\bepsilon{\boldsymbol{\epsilon}}

\def\blambda{\boldsymbol{\lambda}}

\def\bmu{\boldsymbol{\mu}}

\def\bSigma{\boldsymbol{\Sigma}}
\def\btau{\boldsymbol{\tau}}

\def\bOmega{\boldsymbol{\Omega}}

\def\mba{\mathbf{a}}
\def\mbb{\mathbf{b}}
\def\mbc{\mathbf{c}}

\def\mbh{\mathbf{h}}

\def\mbp{\mathbf{p}}

\def\mbr{\mathbf{r}}

\def\mbx{\mathbf{x}}

\def\mbA{\mathbf{A}}
\def\mbB{\mathbf{B}}
\def\mbC{\mathbf{C}}

\def\mbH{\mathbf{H}}
\def\mbI{\mathbf{I}}

\def\mbK{\mathbf{K}}

\def\mbM{\mathbf{M}}

\def\mbO{\mathbf{O}}
\def\mbP{\mathbf{P}}
\def\mbQ{\mathbf{Q}}
\def\mbR{\mathbf{R}}
\def\mbS{\mathbf{S}}

\def\mbU{\mathbf{U}}
\def\mbV{\mathbf{V}}
\def\mbW{\mathbf{W}}
\def\mbX{\mathbf{X}}

\newtheorem{theorem}{Theorem}
\newtheorem{proposition}{Proposition}

\newtheorem{corollary}{Corollary}

\theoremstyle{definition}

\algnewcommand\algorithmicinput{\textbf{Input:}}
\algnewcommand\Input{\item[\algorithmicinput]}
\algnewcommand\algorithmicoutput{\textbf{Output:}}
\algnewcommand\Output{\item[\algorithmicoutput]}
\algnewcommand\algorithmicinit{\textbf{Initialize:}}
\algnewcommand\Init{\item[\algorithmicinit]}

\makeatletter
\newcommand*{\rom}[1]{\expandafter\@slowromancap\romannumeral #1@}
\makeatother

\begin{document}

\title{ORKA: Accelerated Kaczmarz Algorithms for Signal Recovery from One-Bit Samples}
\author{Arian Eamaz, \IEEEmembership{Student Member, IEEE}, Farhang Yeganegi, Deanna Needell, \IEEEmembership{Member, IEEE}, and \\ Mojtaba Soltanalian, \IEEEmembership{Senior Member, IEEE}
\thanks{This work was supported in part by National Science Foundation Grant CCF-1704401. The first two authors contributed equally to this work.}
\thanks{A. Eamaz, F. Yeganegi and M. Soltanalian are with the Department of Electrical and Computer Engineering, University of Illinois Chicago, Chicago, IL 60607, USA (\emph{Corresponding author: Arian Eamaz}).
}
\thanks{D. Needell is with the Department of Mathematics, University of California Los Angeles, Los
Angeles, CA 90095 USA.}
}

\markboth{Submitted to the IEEE TRANSACTIONS ON SIGNAL PROCESSING, 2022 
}
{Shell \MakeLowercase{\textit{et al.}}: Bare Demo of IEEEtran.cls for IEEE Journals}
\maketitle

\begin{abstract}
One-bit quantization with time-varying sampling thresholds
has recently found significant utilization potential in statistical signal processing applications due to its relatively low
power consumption and low implementation cost. In addition
to such advantages, an attractive feature of one-bit analog-to-digital converters (ADCs) is their superior sampling rates as compared to their conventional multi-bit counterparts. This characteristic endows one-bit signal processing frameworks with what we refer to as \emph{sample abundance}.  On the other hand, many signal recovery and optimization problems are formulated as (possibly non-convex) quadratic programs with linear feasibility constraints in the one-bit sampling regime. We demonstrate, with a particular focus on the nuclear norm minimization, that the sample abundance paradigm allows for the transformation of such quadratic problems to merely a linear feasibility problem by forming a large-scale overdetermined linear system; thus removing the need for costly optimization constraints and objectives. To make this achievable, we propose enhanced randomized Kaczmarz algorithms to tackle these highly overdetermined feasibility problems. Several numerical results are presented to illustrate the effectiveness of the proposed
methodologies.  
\end{abstract}


\IEEEpeerreviewmaketitle

\section{Introduction}
\IEEEPARstart{W}{e} consider an optimization problem of the form
\begin{equation}
\label{eq:1nnnn}
\begin{aligned}
\min_{\mbX}\quad &f(\mbX)\\
\text{s.t.} \quad &\mathcal{A}\left(\mbX\right)=\mathbf{y}, \\
&\mbX \in \Omega_{c},
\end{aligned}
\end{equation}
where $f(.)$ is a cost function, $\mbX\in \mathbb{C}^{n_{1}\times n_{2}}$ is the matrix of unknowns, $\mathbf{y}\in \mathbb{R}^{n}$ is the measurement vector, and $\mathcal{A}$ is a linear transformation mapping $\mathbb{C}^{n_{1}\times n_{2}}$ into $\mathbb{R}^{n}$. 

This problem has been used as a relaxed version of some well-known NP-hard problems, and emerging in wide variety of statistical signal processing applications. Although many problems can be expressed in
the form in (\ref{eq:1nnnn}), the applications we will focus on in this paper
include some specific problems of interest in statistical signal processing, which can take advantage of low-resolution (and particularly one bit) sampling and processing:

\begin{itemize}
    \item \textbf{Low-rank matrix recovery}: The task of recovering a low-rank matrix from its linear measurements plays
a central role in computational science. The problem occurs in many areas of applied mathematics such as signal processing \cite{candes2015phase,cai2010singular,davenport20141,davenport2016overview,candes2011tight,chen2013low,li2019survey}, machine learning \cite{haeffele2014structured,nie2012low,blei2012probabilistic,obozinski2010joint,pontile2007convex,argyriou2006multi}, and computer vision\cite{tomasi1992shape}. In this scenario, the cost function of (\ref{eq:1nnnn}),  $f(.)$, is typically to be the \emph{nuclear norm} or the Frobenius norm, and the constraint set $\Omega_{c}$ would be a amplitude restriction on the elements of matrix $\mbX$; see \cite{recht2011null,cai2010singular}.
\item \textbf{Phase retrieval}:
Phase retrieval has received a great deal of interest as it aims to recover an unknown signal solely from phaseless
measurements that depend on the signal through a linear
observation, commanding numerous applications in applied physics and statistical signal processing communities over the past decades \cite{9896984,millane,kim,Fienup93, Krist95, Sarnik,GS1,GS2,fienup1,fienup2,fienup3,candes2013phaselift, candes2014solving}. To have a convex formulation, the phase retrieval has been relaxed into semi-definite programs where the problem boils down to a trace minimization while considering the positive semi-definite constraint \cite{9896984,candes2015phase}.
\item \textbf{Compressed sensing}: Compressed sensing (CS) offers a framework for simultaneous sensing  and compression of finite dimensional vectors, that relies on linear dimensionality reduction. Through a CS formulation, sparse signals may be recovered from highly incomplete measurements\cite{eldar2012compressed}. The problem (1) can be adopted in the CS content when $f\left(\mbX\right)=\left\|\operatorname{vec}\left(\mbX\right)\right\|_{1}$.
\item \textbf{Magnetic resonance imaging}: Reconstructing magnetic resonance images commonly involves collecting a series of frames of data in which a radio frequency excitation produces new transverse magnetization, which is then sampled along a particular trajectory in $k$-sparse representation. Due to meet various physical and physiological constraints, most MRI methods utilize a sequence of acquisitions, each of which partially samples the representation. Let the acquired sequence of measurements be represented by $y_{i}$, where $i$ is the sequence index, and $\left\{\mathcal{A}_{i}\left(\mbX\right)\right\}$ denote a linear transformation, chosen in a manner that promotes sparsity in the range space. In this example, the cost function can be considered to be the $\ell_{1}$-norm, and the sequence of acquisitions are used as linear constraints in (\ref{eq:1nnnn})\cite{lustig2008compressed,fessler2010model}.
\end{itemize}

Sampling the signals of interest at high data rates with high-resolution ADCs would dramatically increase the overall manufacturing cost and power consumption of such ADCs. In multi-bit sampling scenarios, a very large number of quantization levels is necessary in order to represent the original continuous signal in with high accuracy, which in turn leads to a considerable reduction in sampling rate \cite{eamaz2021modified}. This attribute of multi-bit sampling is the key reason for the general emergence of underdetermined systems $n_{1}n_{2}\geq n$ in (\ref{eq:1nnnn}) \cite{candes2013phaselift,candes2015phase,9896984}. An alternative solution to such challenges is to deploy \emph{one-bit quantization} which is an extreme sampling scenario, where the signals are merely compared with given threshold levels at the ADCs, producing sign data ($\pm1$). This enables signal processing equipments to sample at a very high rate, with a considerably lower cost and energy consumption, compared to their counterparts which employ multi-bit ADCs \cite{instrumentsanalog,mezghani2018blind,eamaz2021modified,sedighi2020one}.

In traditional one-bit sampling schemes, the signal recovery is accomplished by comparing the signal with a fixed threshold, usually zero. This creates some difficulties in estimating signal parameters. In contrast, recent works have employed time-varying sampling thresholds, which exhibit enhanced recovery performance for the signal parameters \cite{eamaz2021modified,qian2017admm,gianelli2016one,9896984,eamaz2022covariance,wang2017angular,xi2020gridless}.

In this paper, we consider the deployment of one-bit sampling with time-varying thresholds,leading to an increased sample size and a \emph{highly overdetermined system} as a result. The proposed \emph{O}ne-bit aided \emph{R}andomized \emph{K}aczmarz \emph{Algorithm}, which we refer to as ORKA, can find the desired signal $\mbX^{\star}$ in (\ref{eq:1nnnn}) by (i) generating abundant one-bit measurements, in order to define a large scale overdetermined system where a finite volume feasible set is created for (\ref{eq:1nnnn}), and (ii) solving this obtained linear feasibility problem by leveraging one of the efficient solver families of overdetermined systems, \emph{Kaczmarz algorithms}. The Kaczmarz method \cite{kaczmarz1937angenaherte} is an iterative projection algorithm for solving linear systems of equations and inequalities. It is usually applied to highly overdetermined systems because of its simplicity.
Each iteration projects onto the solution space corresponding to one row in the linear system, in a sequential regimen. The method has been applied to various applications in image reconstruction, digital signal processing, and
computer tomography \cite{feichtinger1995kaczmarz,sezan1987applications,9896984}. Many variants of this iterative method and their convergence rates have been proposed and studied in recent decades for both consistent and inconsistent systems including the randomized Kaczmarz algorithm, the randomized block Kaczmarz algorithm and most recently, the sampling Kaczmarz-Motzkin method \cite{strohmer2009randomized,leventhal2010randomized,needell2014paved,briskman2015block,de2017sampling}.

\subsection{Contributions of the Paper}

In \cite{9896984}, we showed that the sheer number of measurements acquired in one bit sampling facilitates recovering the signal of interest in a less costly manner by making costly constraints such as semidefiniteness and rank redundant. Then, a simple randomized Kaczmarz algorithm (RKA) was utilized to solve the obtained linear feasibility problem. This idea is generalized in this paper to (\ref{eq:1nnnn}) where we generate the abundant samples and eventually introduce a one-bit linear feasibility region named the \emph{one-bit polyhedron}. In other words, by using this technique, we make (\ref{eq:1nnnn}) a large-scale overdetermined system which is the desired application setting for Kaczmarz algorithms.

To solve our highly overdetermined system, we propose two novel variants of RKA which will be compared with the existing RKA variants. Furthermore, an
algorithm is proposed based on our model to adaptively evaluate the time-varying sampling thresholds. The convergence rate of the proposed algorithm is investigated based on the moments generating function of recovery errors and the scaled condition number of the constraint matrix. Finally, the performance of the proposed method is examined in nuclear norm minimization-based problems.

\subsection{Organization of the Paper}

Section~\rom{2} is dedicated to a review of proximal methods which have been utilized to tackle (\ref{eq:1nnnn}) by projecting the final solution on the desired feasible set. In Section~\rom{3}, we will introduce our algorithm to solve (\ref{eq:1nnnn}), ORKA, which tackles the problem as a large-scale overdetermined system and finds the optimal point in the one-bit polyhedron by an accelerated Kaczmarz approach. Moreover, two new variants of the Kaczmarz algorithms are proposed that enhance the convergence rate and the computational complexity of these solvers. To investigate the convergence rate of ORKA, at first, we will introduce a penalty function in Section~\rom{4} based on the Chernoff bound. Section~\rom{5} discusses an iterative algorithm to achieve optimized time-varying sampling threshold sequences which benefit the signal recovery process with enhanced accuracy. As a representative application, in Section~\rom{6}, ORKA and other proposed algorithms will be applied in the  context of low-rank matrix recovery in the form of a nuclear norm minimization problem. Finally, Section~\rom{9} concludes the paper.

\underline{\emph{Notation:}}
We use bold lowercase letters for vectors and bold uppercase letters for matrices. $\mathbb{C}$ and $\mathbb{R}$ represent the set of complex and real numbers, respectively. $(\cdot)^{\top}$ and $(\cdot)^{\mathrm{H}}$ denote the vector/matrix transpose, and the Hermitian transpose, respectively. $\mbI_{N}\in \mathbb{R}^{N\times N}$ is the identity matrix of size $N$. $\operatorname{Tr}(.)$ denotes the trace of the matrix argument. $\left\langle \mbB_{1},\mbB_{2}\right\rangle=\operatorname{Tr}(\mbB_{1}^{\mathrm{H}}\mbB_{2})$ is the standard inner product between two
matrices. The nuclear norm of a matrix $\mbB\in \mathbb{C}^{N_{1}\times N_{2}}$ is denoted $\left\|\mbB\right\|_{\star}=\sum^{M}_{i=1}\sigma_{i}$ where $M$ and $\left\{\sigma_{i}\right\}$ are the rank and singular values of $\mbB$, respectively. The Frobenius norm of a matrix $\mbB$ is defined as $\|\mbB\|_{\mathrm{F}}=\sqrt{\sum^{N_{1}}_{r=1}\sum^{N_{2}}_{s=1}\left|b_{rs}\right|^{2}}$ where $\{b_{rs}\}$ are elements of $\mbB$. The $\ell^{k}$-norm of a vector $\mathbf{b}$ is defined as $\|\mbb\|^{k}_{k}=\sum_{i}|b|^{k}_{i}$. The Hadamard (element-wise) product of two matrices $\mbB_{1}$ and $\mbB_{2}$ is denoted as $\mbB_{1}\odot \mbB_{2}$. Additionally, the  Kronecker product is denoted as $\mbB_{1}\otimes \mbB_{2}$. The vectorized form of a matrix $\mbB$ is written as $\operatorname{vec}(\mbB)$. $\mathbf{1}_{s}$ is the $s$-dimensional all-one vector. Given a scalar $x$, we define $(x)^{+}$ as $\max\left\{x,0\right\}$. $f\asymp g$ means $f$ and $g$ are asymptotically equal. $\operatorname{Diag}\left\{\mathbf{b}\right\}$ denotes a diagonal matrix with $\{b_{i}\}$ as its diagonal elements. 

\section{Projections On Convex Sets: Dealing With Costly Constraints}
\label{sec_1}
To tackle (\ref{eq:1nnnn}), many non-convex and local optimization algorithms have been developed over the years. Nevertheless, in recent decades, convex programming formulations via \emph{relaxation} have come to the fore to approximate \emph{global} solutions. In the convex framework, various iterative methods have been proposed to tackle the problem with a Lagrangian formulation such as \emph{Uzawa's algorithm} and the proximal forward-backward splitting method (PFBS) \cite{daubechies2004iterative,combettes2005signal,cai2010singular}. Moreover, to keep the problem solution inside the constraint set $\Omega_{c}$, the orthogonal projection $\mathcal{P}_{\Omega_{c}}$ is applied to solutions in each iteration. This process is briefly explained below.

The Lagrangian for (\ref{eq:1nnnn}) is written as \cite{cai2010singular},
\begin{equation}
\label{eq:1}
\mathcal{L}\left(\mbX,\blambda\right)= f\left(\mbX\right)+	\left\langle \blambda,\mathbf{y}-\mathcal{A}\left(\mbX\right)\right\rangle,
\end{equation}
where $\blambda\in \mathbb{R}^{n}$. Uzawa's algorithm aims to find a saddle point $\left(\mbX^{\star},\blambda^{\star}\right)$, where $\sup_{\blambda}\inf_{\mbX} \mathcal{L}\left(\mbX,\blambda\right)=\inf_{\mbX}\sup_{\blambda} \mathcal{L}\left(\mbX,\blambda\right)$, with the iterative procedure:
\begin{equation}
\label{eq:2}
\begin{aligned}
\begin{cases}\mathcal{L}\left(\mbX^{k},\blambda^{k-1}\right)=\min_{\mbX}\mathcal{L}\left(\mbX,\blambda^{k-1}\right), \\ \blambda^{k}=\mathcal{P}_{\Omega_{c}}\left(\blambda^{k-1}+\alpha_{k}\left(\mathbf{y}-\mathcal{A}\left(\mbX^{k}\right)\right)\right),\end{cases}
\end{aligned}
\end{equation}
where $\alpha_{k}$ is the step size. This iterative steps can be rewritten as
\begin{equation}
\label{eq:3}
\begin{aligned}
\begin{cases}\mbX^{k}=\operatorname{Prox}_{f}\left(\mathcal{A}^{\star}\left(\blambda^{k-1}\right)\right), \\ \blambda^{k}=\mathcal{P}_{\Omega_{c}}\left(\blambda^{k-1}+\alpha_{k}\left(\mathbf{y}-\mathcal{A}\left(\mbX^{k}\right)\right)\right),\end{cases}
\end{aligned}
\end{equation}
where $\operatorname{Prox}_{f}$ is the proximal operator minimizing the Lagrangian function, and $\mathcal{A}^{\star}$ is the adjoint of $\mathcal{A}$.

Since every linear equation can be reformulated in standard form, we recast $\mathcal{A}\left(\mbX\right)=\mathbf{y}$ as $\mbA\mathbf{x}=\mathbf{y}$, where $\mbA\in \mathbb{C}^{n\times n_{1}n_{2}}$ is a matrix version of the operator $\mathcal{A}$, and $\mathbf{x}=\operatorname{vec}\left(\mbX\right)$\cite{recht2011null}. 
The optimization problem (\ref{eq:1nnnn}) is equivalently given by \cite{cai2010singular,candes2015phase}
\begin{equation}
\label{eq:4}
\begin{aligned}
\min_{\mbX}\quad &g\left(\mbX\right)=\frac{1}{2}\left\|\mathbf{y}-\mbA\operatorname{vec}\left(\mbX\right)\right\|^{2}_{2}+\lambda f(\mbX) \\
\text{s.t.} \quad
&\mbX \in \Omega_{c}.
\end{aligned}
\end{equation}
To solve this problem, instead of using proximal methods, a projected gradient method such as Nesterov iterative approach may be utilized, i.e., 
$\mbX^{k}=\mathcal{P}_{\Omega_{c}}\left(\mbX^{k-1}-\alpha_{k}\nabla g\left(\mbX^{k-1}\right)\right)$.

Famous examples for $\operatorname{Prox}_{f}$ and $\mathcal{P}_{\Omega_{c}}$, are the singular value thresholding operator (SVT) and the semi-definite orthogonal projector, respectively. SVT is useful when $f\left(\mbX\right)=\left\|\mbX\right\|_{\star}$, mathematically defined as \cite{cai2010singular}:
\begin{equation}
\label{eq:5}
\begin{aligned}
\mathcal{D}_{\delta} = \mbU \operatorname{Diag}\left\{\left(\sigma_{k}-\delta\right)^{+}\right\}\mbV^{\top},
\end{aligned}
\end{equation}
where $\mbU$ and $\mbV$ are unitary matrices from singular value decomposition (SVD), and $\{\sigma_{k}\}$ are the singular values. Furthermore, the semi-definite projector emerges in semi-definite programming where the convex constraint set is a positive semi-definite (PSD) matrix. It compares eigenvalues of the solution in each iteration with zero or a fixed threshold \cite{candes2015phase}, i.e., 
\begin{equation}
\label{eq:5ar}
\begin{aligned}
\mathcal{P}_{\Omega_{c}} = \mbU  \operatorname{Diag}\left\{\left(\lambda_{k}-\delta\right)^{+}\right\}\mbU^{\top},
\end{aligned}
\end{equation}
where $U$ is the unitary matrix coming from the Schur decomposition. In the case of both operators, the approximate solution should be projected onto a feasible convex set at each iteration via recovering all singular values and eigenvalues and comparing their smaller elements with a threshold, which is quite expensive \cite{candes2015phase}. 

An interesting alternative to enforcing the feasible set $\mathcal{F}_{\mbX}=\left\{\operatorname{Prox_{f}\cap \Omega_{c}}\right\}$ in (\ref{eq:1nnnn}) emerges when one increases the number of samples $n$, and solves the overdetermined linear system of equations with $n\geq n_{1}n_{2}$. In this sample abundance regimen, the linear constraint $\mathcal{A}\left( \mbX\right)=\mathbf{y}$ may actually yield the optimum inside $\mathcal{F}_{\mbX}$. As a result of increasing the number of samples, it is possible that the intersection of these hyperplanes will achieve the optimal point without the need to consider costly constraints. However, this idea may face practical limitations in the case of multi-bit quantization systems since ADCs capable of ultra-high rate sampling are difficult and expensive to produce. Moreover, one cannot necessarily expect these constraints to intersect with $\mathcal{F}_{\mbX}$ in such a way to form a finite-volume space before the optimum is obtained \cite{9896984,candes2013phaselift}.

In the next section, by deploying the idea of one-bit sampling with time-varying thresholds, linear equality constraints are superseded by a massive array of linear inequalities in forming the feasible polyhedron. Therefore, by increasing the number of samples, a finite-volume space may be created inside $\mathcal{F}_{\mbX}$ with shrinking size; making projections on $\Omega_{c}$ redundant. \emph{From a practical point of view, one-bit sampling is done efficiently at a very high rate with a significantly lower cost compared to its high-resolution counterpart. It has been examined in \cite{9896984} that even though only partial information is made available to one-bit signal processing algorithms, they can achieve acceptable recovery performance with less complexity compared to the high-resolution scenario. Thus, it is both practical and necessary to study the ground-breaking opportunities that emerge in the context of the wide array of problems formulated as (\ref{eq:1nnnn}) due to the availability of a large number of one-bit samples.}

\section{Proposed Algorithm}
\label{sec_3}
In this section, at first we begin by presenting a summarized review of randomized Kaczmarz algorithms. Then, we propose a novel Kaczmarz method variant formulated based on the \emph{sampling Kaczmarz-Motzkin algorithm} (SKM) and a \emph{preconditioning approach}. One-bit sampling via time-varying thresholds will be combined with the proposed randomized Kaczmarz method to create highly overdetermined linear inequalities. This paves the way for the recovery of the desired signal $\mbX^{\star}$ in (\ref{eq:1nnnn}) without solving the original optimization problem; merely by tacking accounts of its linear constraints. We name our algorithm \emph{O}ne-bit aided \emph{R}andomized \emph{K}aczmarz \emph{A}lgorithm (ORKA). Due to the block structure of the linear feasibility in ORKA, we will propose a block-based Kaczmarz algorithm accordingly.

\subsection{Randomized Kaczmarz Algorithm (RKA)}
\label{sec_RKA}

The randomized Kaczmarz algorithm (RKA) is a \emph{sub-conjugate gradient method} to solve a linear feasibility problem, i.e, $\mbC\mathbf{x}	\preceq\mathbf{b}$ where $\mbC$ is a ${m\times n}$ matrix with $m>n$ \cite{leventhal2010randomized,strohmer2009randomized}. Conjugate-gradient methods immediately turn the mentioned inequality to an equality in the following form:
\begin{equation}
\label{bikhod}
\left(\mbC\mathbf{x}-\mathbf{b}\right)^{+}=0,
\end{equation}
and then, approach the solution by the same process as used for systems of equations. Without any loss of generality,
consider (\ref{bikhod}) to be a polyhedron:
\begin{equation}
\label{eq:21}
\begin{aligned}
\begin{cases}\mbc_{j} \mathbf{x} \leq b_{j} & \left(j \in I_{\leq}\right), \\ \mbc_{j} \mathbf{x}=b_{j} & \left(j \in I_{=}\right),\end{cases}
\end{aligned}
\end{equation}
where the disjoint index sets $I_{\leq}$ and $I_{=}$ partition our sample index set $\mathcal{J}$, and $\{\mbc_{j}\}$ denote the rows of $\mbC$. Based on this problem, the projection coefficient $\beta_{i}$ of the RKA is defined as \cite{leventhal2010randomized,briskman2015block,dai2013randomized}:
\begin{equation}
\label{eq:22}
\beta_{i}= \begin{cases}\left(\mbc_{j} \mathbf{x}_{i}-b_{j}\right)^{+} & \left(j \in I_{\leq}\right), \\ \mbc_{j} \mathbf{x}_{i}-b_{j} & \left(j \in I_{=}\right).\end{cases}
\end{equation}
Also, the unknown column vector $\mathbf{x}$ is iteratively updated as
\begin{equation}
\label{eq:23}
\mathbf{x}_{i+1}=\mathbf{x}_{i}-\frac{\beta_{i}}{\left\|\mbc_{j}\right\|^{2}_{2}} \mbc^{\mathrm{H}}_{j},
\end{equation}
where, at each iteration $i$, the index $j$ is chosen independently at random from the set $\mathcal{J}$, following the distribution
\begin{equation}
\label{eq:24}
P\{j=k\}=\frac{\left\|\mbc_{k}\right\|^{2}_{2}}{\|\mbC\|_{\mathrm{F}}^{2}}.
\end{equation}

If the system (\ref{eq:21}) is consistent and its feasible region is nonempty, RKA converges linearly in expectation \cite{strohmer2009randomized,leventhal2010randomized}:
\begin{equation}
\label{eq:15}
\mathbb{E}\left\{\hbar\left(\mathbf{x}_{i},\mathbf{x}^{\star}\right)\right\} \leq q^{i}~ \hbar\left(\mathbf{x}_{0},\mathbf{x}^{\star}\right),
\end{equation}
where $\hbar\left(\mathbf{x}_{i},\mathbf{x}^{\star}\right)=\left\|\mathbf{x}_{i}-\mathbf{x}^{\star}\right\|_{2}^{2}$ is the distance function between two points in the space, $\mathbf{x}^{\star}$ is a desired point, $i$ is the number of required iterations for RKA, and $q \in \left(0,1\right)$ is given as
\begin{equation}
\label{eq:16}
q=1-\frac{1}{\kappa^{2}\left(\mbC\right)},
\end{equation}
with $\kappa\left(\mbC\right)=\|\mbC\|_{\mathrm{F}}\|\mbC^{\dagger}\|_{2}$ denoting the scaled condition number.
\subsection{Sampling Kaczmarz-Motzkin Algorithm (SKM)}
\label{sec_SKM}
The SKM combines the ideas of both the RKA and the Motzkin method. Its generalized convergence theorem, and a validation of feasibility, which has been formulated based on the convergence analysis of RKA and sampling Motzkin method for solving linear feasibility problem have been fully explored in \cite{de2017sampling}. 

The central contribution of SKM lies in its innovative way of projection plane selection. The hyperplane selection is done as follows. At iteration $i$ the SKM algorithm selects a collection of $\gamma$ (denoted by the set $\tau_{i}$), uniformly at random out of $m$ rows of the constraint matrix $\mbC$. Then, out of these $\gamma$ rows, the row with maximum positive residual is selected.  Finally, the solution is updated as \cite{de2017sampling,sarowar2020sampling}:
\begin{equation}
\label{eq:230}
\mathbf{x}_{i+1}=\mathbf{x}_{i}-\lambda_{i}\frac{\beta_{i}}{\left\|\mbc_{j^\star}\right\|^{2}_{2}} \mbc^{\mathrm{H}}_{j^{\star}},
\end{equation}
where $j^{\star}=\operatorname{argmax}~\left\{ \left(\mbc_{j} \mathbf{x}_{i}-b_{j}\right)^{+}\right\},~j\in\tau_{i}$, and $\lambda_{i}$ is a relaxation parameter which for consistent systems must satisfy \cite{strohmer2009randomized},
\begin{equation}
\label{relaxation} 
0\leq \lim_{i\to\infty} \inf \lambda_{i}\leq \lim_{i\to\infty} \sup \lambda_{i}<2,
\end{equation}
to ensure convergence. 
The convergence bound for SKM is given by
\begin{equation}
\label{eq:150}
\mathbb{E}\left\{\hbar\left(\mathbf{x}_{i},\mathbf{x}^{\star}\right)\right\} \leq \left(1-\frac{2\lambda_{i}-\lambda^{2}_{i}}{\kappa^{2}\left(\mbC\right)}\right)^{i}~ \hbar\left(\mathbf{x}_{0},\mathbf{x}^{\star}\right).
\end{equation}
In the case where the constraint matrix is normalized, i.e. $\|\mbc_{j}\|^{2}_{2}=1$, $s_{i}$ is the number of satisfied constraints after iteration $i$, and $L_{i}=\max\left\{m-s_{i},m-\gamma\right\}$, for the $(i+1)$th iteration we have \cite{de2017sampling},
\begin{equation}
\label{eq:1500}
\mathbb{E}\left\{\hbar\left(\mathbf{x}_{i},\mathbf{x}^{\star}\right)\right\} \leq \left(1-\frac{\sigma^{2}_{min}\left(2\lambda_{i}-\lambda^{2}_{i}\right)}{V_{i}}\right)^{i}~ \hbar\left(\mathbf{x}_{0},\mathbf{x}^{\star}\right).
\end{equation}
This recovery error bound is tighter than (\ref{eq:150}).
\subsection{Our Contribution: Preconditioned SKM (PrSKM)}
\label{sec_preSKM}
According to the convergence rate formula of RKA, if we can reduce the value of the scaled condition number, the convergence is accelerated, and the upper bound of the recovery error $\mathbb{E}\left\{\left\|\mathbf{x}_{i}-\mathbf{x}^{\star}\right\|_{2}^{2}\right\}$ decreases. Moreover, by having a lower value of $q$, a lower number of iterations is required to achieve a specific recovery error bound, usually considered to be the algorithm's termination criterion. Consequently, let $I$ is the number of iterations, the computational cost of RKA which behaves as $\mathcal{O}\left(I n\right)$, is diminished as well. To make this happen, one can start from reducing $q=1-\frac{1}{\kappa^{2}\left(\mbC\right)}$ which occurs when the scaled condition number $\kappa\left(\mbC\right)$ is diminished; a condition that can be satisfied by considering the following theorem.
\begin{theorem}
\label{scaled_number}
The infimum scaled condition number of a matrix $\mbC\in\mathbb{R}^{m\times n}$ is given by
\begin{equation}
\label{eq:Arian}
\inf_{\mbC} \kappa\left(\mbC\right) = \sqrt{n},
\end{equation}
which is achieved if and only if $\mbC$ is of the form $\mbC=\alpha \mbU$, where $\mbU$ is an orthonormal matrix and $\alpha\in\mathbb{R}$ is a scalar.
\end{theorem}
\begin{proof}
The condition number of the matrix $\mbC$ is defined as $\varrho\left(\mbC\right)=\frac{\sigma_{max}}{\sigma_{min}}$, where $\sigma_{max}$ and $\sigma_{min}$ are its minimum and maximum singular values, respectively \cite{van1996matrix}. The scaled condition number can be written as $\kappa\left(\mbC\right)=\frac{\|\mbC\|_{\mathrm{F}}}{\sigma_{min}}$. Therefore, the scaled condition number has the following relation with $\varrho\left(\mbC\right)$:
\begin{equation}
\label{gorang}
\kappa\left(\mbC\right)=\frac{\|\mbC\|_{\mathrm{F}}}{\sigma_{max}}\varrho\left(\mbC\right). 
\end{equation}
Furthermore, the condition number can be considered to be an upper bound for the scaled condition number as well based on the readily-known inequality relation between norm-$2$ and the Frobenius norm \cite{van1996matrix}:
\begin{equation}
\label{sca}
\begin{aligned}
\|\mbC\|_{\mathrm{F}}&\leq \sqrt{n}  \|\mbC\|_{\mathrm{2}},\\
\frac{\|\mbC\|_{\mathrm{F}}}{\sigma_{\min}} &\leq \sqrt{n} \frac{\|\mbC\|_{\mathrm{2}}}{\sigma_{\min}},
\end{aligned}
\end{equation}
or equivalently,
\begin{equation}
\label{sca1}
\begin{aligned}
\kappa\left(\mbC\right)\leq \sqrt{n} \varrho\left(\mbC\right).
\end{aligned}
\end{equation}
Thus, lowering $\varrho\left(\mbC\right)$ generally leads to a decreasing scaled condition number. Additionally,
the lowest possible value for $\varrho$ is $1$ which is achieved for scaled unitary matrices $\mbU$ as if we let $\mbS=\alpha \mbU$, and $\mbO=\mbS^{\top}\mbS=\alpha^{2}\mbI_{n}$, then, $\sigma_{min}=\sigma_{max}=\alpha$, and $\varrho=1$. Vice versa, if $\varrho=1$, it means $\sigma_{min}=\sigma_{max}$ which leads to a diagonal matrix $\mbO=\alpha^{2}\mbI_{n}$. It is straightforward to verify that the decomposition of $\mbO$ results in an $\mbS$ that is a scaled-version of an orthonormal matrix. As a result, the lowest achievable upper bound for the scaled condition number is obtained as $\kappa\left(\mbC\right)\leq \sqrt{n}$, and according to (\ref{gorang}), $ \kappa\left(\mbC\right)=\frac{\alpha\|\mbU\|_{\mathrm{F}}}{\alpha}=\sqrt{n}$. 
\end{proof}

Accordingly, it would be enough to make our matrix $\mbC$ unitary by a process which is referred to as the \emph{preconditioning method}. In preconditioning, the linear feasibility is rewritten as
\begin{equation}
\label{new}
\mbC \mbM \mathbf{z}\preceq \mathbf{b},
\end{equation}
where $\mbM$ is the preconditioner and $\mathbf{x}=\mbM \mathbf{z}$. The straightforward way to approach this task is to use QR decomposition where the constraint matrix is decomposed as $\mbC=\mbQ_{c}\mbR_{c}$, with unitary $\mbQ_{c}\in\mathbb{R}^{m\times n}$, and $\mbR_{c}\in\mathbb{R}^{n\times n}$ is an upper triangular matrix, leading to
\begin{equation}
\label{eq:qr} 
\mbQ_{c}=\mbC\mbR^{-1}. 
\end{equation}
Thus, the good choice for the preconditioner is $\mbM=\mbR^{-1}_{c}$. To find the desired point $\mathbf{z}^{\star}$, the SKM is selected in order to apply to the linear feasibility (\ref{new}), then the desired signal $\mathbf{x}^{\star}$ is obtained from $\mathbf{x}^{\star}=\mbR^{-1}_{c}\mathbf{z}^{\star}$. We refer to this method \emph{Pr}econditioned \emph{SKM} (PrSKM).

\begin{proposition}[PrSKM]
\label{Pr_SKM}
The proposed algorithm, PrSKM, can be summarized as follows:
\begin{enumerate}
    \item Calculate the QR decomposition of the constraint matrix $\mbC$ to obtain the preconditioner $\mbM$.
    \item Using the change of variables, $\mathbf{x}=\mbM\mathbf{z}$, obtain $\mbH \mathbf{z}\preceq \mathbf{b}$, where $\mbH=\mbC\mbM$ and $\mbM=\mbR^{-1}_{c}$.
    \item Choose a sample set of $\gamma$ constraints (denoted as $\tau_{i}$) uniformly at random from the rows of $\mbH$.
    \item From these $\gamma$ constraints, choose $j^{\star}=\operatorname{argmax}~\left\{ \left(\mbh_{j} \mathbf{z}_{i}-b_{j}\right)^{+}\right\},~j\in\tau_{i}$ where $\mbh_{j}$ is the $j$th row of $\mbH$.
    \item Update the solution via the iterations $\mathbf{z}_{i+1}=\mathbf{z}_{i}-\lambda_{i}\frac{\left(\mbh_{j^{\star}} \mathbf{z}_{i}-b_{j^{\star}}\right)^{+}}{\left\|\mbh_{j^\star}\right\|^{2}_{2}} \mbh^{\mathrm{H}}_{j^{\star}}$ until convergence.
    \item Recover the desired signal from the final solution of SKM as $\mathbf{x}^{\star}=\mbM\mathbf{z}^{\star}$.
\end{enumerate}
\end{proposition}
The scaled condition number of PrSKM is obtained as $\kappa\left(\mbH\right)=\sqrt{n}$, which implies $q=\frac{n-1}{n}$. 

\subsection{One-Bit Polyhedron}
\label{one-bit}
Consider a bandlimited signal $y\in L^{2}$ to be represented by its samples via the standard sampling formula\cite{bhandari2020unlimited},
\begin{equation}
\label{eq:1shanon}
0<\mathrm{T} \leqslant \frac{\pi}{\Omega}, \quad y(t)=\sum_{k=-\infty}^{k=+\infty} y(k \mathrm{T}) \operatorname{sinc}\left(\frac{t}{\mathrm{T}}-k\right),
\end{equation}
where $1/\mathrm{T}$ is the sampling rate and $\operatorname{sinc}(t)=\frac{\sin (\pi t)}{(\pi t)}$ is an \emph{ideal} low-pass filter. Suppose $y_{k}=y(k\mathrm{T})$ denotes the uniform samples of $y(t)$ with the sampling rate $\mathrm{T}$. Let $r_{k}$ denote the quantized version of $y[k]$ with the formulation
\begin{equation}
\label{eq:2bit}
r_{k} = Q(y_{k}),
\end{equation}
where $Q$ denotes the quantization effect. Consider a non-zero time-varying Gaussian threshold $\boldsymbol{\uptau}=\left[\uptau_{k}\right]$ with the distribution $\boldsymbol{\uptau} \sim \mathcal{N}\left(\mathbf{d}=\mathbf{1}d,\bSigma\right)$. In the case of one-bit quantization with such time-varying sampling thresholds, (\ref{eq:2bit}) is simply written as
\begin{equation}
\label{eq:3bit}
r_{k} = \operatorname{sgn}\left(y_{k}-\tau_{k}\right),
\end{equation}
where $\operatorname{sgn}(.)$ is the sign function.
\begin{algorithm}[t]
\caption{Block SKM Algorithm}
\label{algorithm_1}
\begin{algorithmic}[1]
\Statex\emph{Input:} Matrix $\mbB\in \mathbb{R}^{mn\times d}$ where $\mbB=\left[\begin{array}{c|c|c}\mbB^{\top}_{1} &\cdots&\mbB^{\top}_{m}\end{array}\right]^{\top}$, right-hand side $\mbb$ with dimension $mn$, initial value of $\mathbf{x}_{0}$ with dimension $n$, convergence tolerance $\epsilon>0$, and relaxation parameter $\lambda_{i}$.
\Statex \emph{Output:} An estimate $\mathbf{x}$ for the solution to the linear feasibility problem $\mbB\mathbf{x}\preceq \mathbf{b}$.
\State Initiate the following loop by setting $i=0$.
\While{$\left\|\left(\mbB\mathbf{x}_{i}-\mbb\right)^{+}\right\|_{2}\leq \epsilon$}
\State Choose a block $\mbB_{j}$ uniformly at random with the probability $P\{j=k\}=\frac{\left\|\mbB_{k}\right\|^{2}_{\mathrm{F}}}{\|\mbB\|_{\mathrm{F}}^{2}}$.
\State Let $\mathbf{e}^{\prime}$ denote the sorted version of $\mathbf{e}$ from $e_{\text{max}}$ (the maximum element of $\mathbf{e}$) to $e_{\text{min}}$ (the minimum element of $\mathbf{e}$).
\State Select the first $k^{\prime}<d$ element of $\mathbf{e}^{\prime}$ and construct the problem $\mbB_{j}^{\prime}\mathbf{x}	\preceq\mathbf{b}_{j}^{\prime}$, where $\mbB_{j}^{\prime}\in\mathbb{R}^{k^{\prime}\times d}$ and $\mathbf{b}_{j}^{\prime}\in\mathbb{R}^{k^{\prime}\times 1}$.
\State Compute the Moore-Penrose of $\mbB_{j}^{\prime}$, i.e.,
\[
\mbB_{j}^{\prime\dagger}\gets\mbB_{j}^{\prime\top}\left(\mbB_{j}^{\prime}\mbB_{j}^{\prime\top}\right)^{-1}.
\]
\State Update the solution $\mbx_{i+1}$ as:
\[
\mathbf{x}_{i+1}\gets\mathbf{x}_{i}-\lambda_{i}\mbB_{j}^{\prime\dagger}\left(\mbB_{j}^{\prime}\mathbf{x}-\mathbf{b}_{j}^{\prime}\right)^{+}.
\]
\State Increase $i$ by one.
\EndWhile
\end{algorithmic}
\end{algorithm}
The information gathered through the one-bit sampling with time-varying thresholds presented here may be formulated in terms of an overdetermined linear system of inequalities. In Eq.~(\ref{eq:3bit}), 
\begin{equation}
\label{eq:220000}
r_{k}=\begin{cases} +1 &y_{k}>\uptau_{k}, \\ -1&y_{k}<\uptau_{k}.\end{cases}
\end{equation}
Therefore, one can formulate the geometric location of the signal as
\begin{equation}
\label{eq:4bit}
r_{k}\left(y_{k}-\uptau_{k}\right) \geq 0.
\end{equation}
Let $\mathbf{y}=[y_{k}]$ and $\mbr=[r_{k}]$. Then, the vectorized representation of (\ref{eq:4bit}) is
\begin{equation}
\label{eq:5arar}
\mbr\odot\left(\mathbf{y}-\boldsymbol{\uptau}\right) \geq \mathbf{0},
\end{equation}
or equivalently
\begin{equation}
\label{eq:6}
\begin{aligned}
\bOmega \mathbf{y} &\succeq \mbr \odot \boldsymbol{\uptau},
\end{aligned}
\end{equation}
where $\bOmega \triangleq \operatorname{diag}\left\{\mbr\right\}$. Suppose $\mathbf{y},\boldsymbol{\uptau} \in \mathbb{R}^{n}$, and that $\boldsymbol{\uptau}^{(\ell)}$ denotes the time-varying sampling threshold sequence in $\ell$-th experiment where  $\ell\in\mathcal{L}=\{1,\cdots,m\}$. According to (\ref{eq:6}), we have
\begin{equation}
\label{eq:7}
\begin{aligned}\begin{cases}
\bOmega^{(1)} \mathbf{y} \succeq \mbr^{(1)} \odot \boldsymbol{\uptau}^{(1)}&\quad \mbr^{(1)} = \operatorname{sgn}\left(\mathbf{y}-\boldsymbol{\uptau}^{(1)}\right),\\
\bOmega^{(2)} \mathbf{y} \succeq \mbr^{(2)} \odot \boldsymbol{\uptau}^{(2)}&\quad \mbr^{(2)} = \operatorname{sgn}\left(\mathbf{y}-\boldsymbol{\uptau}^{(2)}\right),\\
\vdots&\quad \vdots\\
\bOmega^{(m)} \mathbf{y} \succeq \mbr^{(m)} \odot \boldsymbol{\uptau}^{(m)}&\quad \mbr^{(m)} = \operatorname{sgn}\left(\mathbf{y}-\boldsymbol{\uptau}^{(m)}\right),\end{cases}
\end{aligned}
\end{equation}
where $\bOmega^{(\ell)}=\operatorname{diag}\left\{\mbr^{(\ell)}\right\}$. In Eq.~(\ref{eq:7}), we have $m$ linear system of inequalities which can also be merged in one inequality as
\begin{equation}
\label{eq:8}
\Tilde{\bOmega} \mathbf{y} \succeq \operatorname{vec}\left(\mbR\right)\odot \operatorname{vec}\left(\bGamma\right),
\end{equation}
where $\mbR$ and $\bGamma$ are matrices with $\left\{\mbr^{(\ell)}\right\}$ and $\left\{\boldsymbol{\uptau}^{(\ell)}\right\}$ representing their columns, respectively, and $\Tilde{\bOmega}$ is
\begin{equation}
\label{eq:9}
\Tilde{\bOmega}=\left[\begin{array}{c|c|c}
\bOmega^{(1)} &\cdots &\bOmega^{(m)}
\end{array}\right]^{\top}, \quad \Tilde{\bOmega}\in \mathbb{R}^{m n\times n}.
\end{equation}
Assuming a large number of samples which is a common situation in one-bit sampling scenarios, hereafter, we consider (\ref{eq:8}) as an overdetermined linear system of inequalities associated with the sampling scheme presented in (\ref{eq:3bit}).
The inequality (\ref{eq:8}) can be recast by a polyhedron as
\begin{equation}
\label{eq:8n}
\begin{aligned}
\mathcal{P}_{\mathbf{y}} = \left\{\mathbf{y} \mid \Tilde{\bOmega} \mathbf{y}\succeq \operatorname{vec}\left(\mbR\right)\odot \operatorname{vec}\left(\bGamma\right)\right\},
\end{aligned}
\end{equation}
which we refer to as the \emph{one-bit polyhedron}.
\subsection{ORKA: Towards Circumventing Costly Constraints }
\label{ORKA}
If one applies one-bit sampling with time-varying sampling thresholds to the measurement $\mathbf{y}\in\mathbb{R}^{n}$ from (\ref{eq:1nnnn}) following the process defined in Subsection~\ref{one-bit}, the arising inequality system is simply given by
\begin{equation}
\label{eq:8over}
\Tilde{\bOmega}\mbA\mathbf{x} \succeq \operatorname{vec}\left(\mbR\right)\odot \operatorname{vec}\left(\bGamma\right),
\end{equation}
where $\mbA\mathbf{x}=\mathbf{y}$, and $\mathbf{x}=\operatorname{vec}\left(\mbX\right)$. Consequently, the one-bit polyhedron for this problem is obtained as
\begin{equation}
\label{eq:80n}
\begin{aligned}
\mathcal{P}_{\mathbf{x}} = \left\{\mathbf{x} \mid \mbP \mathbf{x} \succeq \operatorname{vec}\left(\mbR\right)\odot \operatorname{vec}\left(\bGamma\right)\right\},
\end{aligned}
\end{equation}
where $\mbP=\Tilde{\bOmega}\mbA$.

By taking advantage of one-bit sampling, in the asymptotic scenario of with sample abundance, the space restricted by the one-bit polyhedron $\mathcal{P}_{\mathbf{x}}$, \emph{shrinks} to become contained inside the feasible set $\mathcal{F}_{\mbX}$. Note that this shrinking space always contains the global minima, with a volume that is diminished with an increased sample size. As a result, instead of using proximal operators and orthogonal projectors, it is enough to find the desired signal $\mathbf{x}^{\star}$ in (\ref{eq:80n}).
To do so, one can use the PrSKM algorithm proposed in Section~\ref{sec_preSKM}. It is worth noting that the PrSKM is a \emph{row-based} algorithm, where at each iteration, the row index is chosen independently at random. However, the matrix $\mbP$ in (\ref{eq:80n}) has a block structure with the following formulation
\begin{equation}
\label{eq:90}
\mbP=\left[\begin{array}{c|c|c}
\mbA^{\top}\bOmega^{(1)} &\cdots &\mbA^{\top}\bOmega^{(m)}
\end{array}\right]^{\top}, \quad \mbP\in\mathbb{R}^{mn\times d},
\end{equation}
where $d=n_{1}n_{2}$. Therefore, it is useful to investigate the block-based RKA methods to find the desired signal in $\mathcal{P}_{\mathbf{x}}$ for further efficiency enhancement. Our proposed algorithm, \emph{Block SKM}, is described as follows.

\begin{algorithm}[t]
\caption{Architecture of ORKA.}
\label{algorithm_2}
\begin{algorithmic}[1]
\Statex \emph{Input:} The measurement vector $\mathbf{y}$ obtained as $\mathcal{A}\left(\mbX\right)=\mathbf{y}$ from (\ref{eq:1nnnn}), $m$ sequences of time-varying sampling thresholds generated as $\left\{\boldsymbol{\uptau}^{(\ell)}\sim \mathcal{N}\left(\mathbf{0},\mbI\right); \ell\in\mathcal{L}\right\}$.
\Statex \emph{Output:} Recovered optimal signal $\mathbf{x}^{\star}=\operatorname{vec}\left(\mbX^{\star}\right)$.
\State Apply one-bit sampling on $\mathbf{y}$ and generate sequences of one-bit measurements from:
\[
\mbr^{(\ell)}\gets\operatorname{sgn}\left(\mathbf{x}-\boldsymbol{\uptau}^{(\ell)}\right), \quad \ell \in \mathcal{L}.
\]
\State Construct a linear feasibility region from the one-bit sampled data as:
\[
\bOmega^{(\ell)} \mathbf{y} \succeq \mbr^{(\ell)} \odot \boldsymbol{\uptau}^{(\ell)}, \quad \ell \in \mathcal{L}.
\]
\State Define a highly-overdetermined system, the one-bit polyhedron, based on obtained inequalities:
\[
\mathcal{P}_{\mathbf{x}} = \left\{\mathbf{x} \mid \mbP \mathbf{x} \succeq \operatorname{vec}\left(\mbR\right)\odot \operatorname{vec}\left(\bGamma\right)\right\},
\]
where $\mbR$ and $\bGamma$ are matrices with $\left\{\mbr^{(\ell)}\right\}$ and $\left\{\boldsymbol{\uptau}^{(\ell)}\right\}$ representing their columns, respectively.
\State Employ RKA variants (PrSKM or Block SKM) to recover $\mbX^{\star}$ within the one-bit polyhedron.
\end{algorithmic}
\end{algorithm}

\begin{proposition}[Block SKM]
\label{Block_SKM}
We have a linear feasibility problem $\mbB\mathbf{x}\preceq \mathbf{b}$ where $\mbB=\left[\begin{array}{c|c|c}\mbB^{\top}_{1} &\cdots&\mbB^{\top}_{m}\end{array}\right]^{\top}$, and $\mathbf{b}=\left[\begin{array}{c|c|c}\mathbf{b}^{\top}_{1} &\cdots&\mathbf{b}^{\top}_{m}\end{array}\right]^{\top}$. The proposed algorithm for feasible signal recovery, Block SKM, can be summarized as follows:
\begin{enumerate}
    \item Choose a block $\mbB_{j}$ uniformly at random with the probability $P\{j=k\}=\frac{\left\|\mbB_{k}\right\|^{2}_{\mathrm{F}}}{\|\mbB\|_{\mathrm{F}}^{2}}$.
    \item Compute $\mathbf{e}=\mbB_{j}\mathbf{x}-\mathbf{b}_{j}$.
    \item Let $\mathbf{e}^{\prime}$ denote the sorted version of $\mathbf{e}$ from $e_{\text{max}}$ (the maximum element of $\mathbf{e}$) to $e_{\text{min}}$ (the minimum element of $\mathbf{e}$). This step is inspired by the idea of the Motzkin sampling, presented in Section~\ref{sec_SKM}, to have an accelerated convergence.
    \item Select the first $k^{\prime}<d$ element of $\mathbf{e}^{\prime}$ and construct the sub-problem $\mbB_{j}^{\prime}\mathbf{x}	\preceq\mathbf{b}_{j}^{\prime}$, where $\mbB_{j}^{\prime}\in\mathbb{R}^{k^{\prime}\times d}$ and $\mathbf{b}_{j}^{\prime}\in\mathbb{R}^{k^{\prime}\times 1}$. The reason behind choosing $k^{\prime}<d$ is due to the computation of $\left(\mbB_{j}^{\prime}\mbB_{j}^{\prime\top}\right)^{-1}$ in the next step (Step $5$). For $k^{\prime}>d$, the matrix $\mbB_{j}^{\prime}\mbB_{j}^{\prime\top}$ is rank-deficient and its inverse is not available.
    \item Compute the Moore-Penrose of $\mbB_{j}^{\prime}$, $\mbB_{j}^{\prime\dagger}=\mbB_{j}^{\prime\top}\left(\mbB_{j}^{\prime}\mbB_{j}^{\prime\top}\right)^{-1}$.
    \item Update the solution $\mathbf{x}_{i+1}=\mathbf{x}_{i}-\lambda_{i}\mbB_{j}^{\prime\dagger}\left(\mbB_{j}^{\prime}\mathbf{x}-\mathbf{b}_{j}^{\prime}\right)^{+}$. This update process is inspired from the randomized block Kaczmarz method \cite{elfving1980block,needell2014paved} which takes advantage of the efficient matrix-vector multiplication, thus giving the method a significant reduction in computational cost \cite{briskman2015block}.
\end{enumerate}
\end{proposition}
The steps of the block SKM and ORKA are summarized in Algorithm~\ref{algorithm_1} and Algorithm~\ref{algorithm_2}, respectively. To examine the performance of the block SKM, we compare it to the PrSKM, SKM and RKA.
\begin{figure*}[t]
	\centering
	\begin{subfigure}[b]{0.45\textwidth}
		\includegraphics[width=1\linewidth]{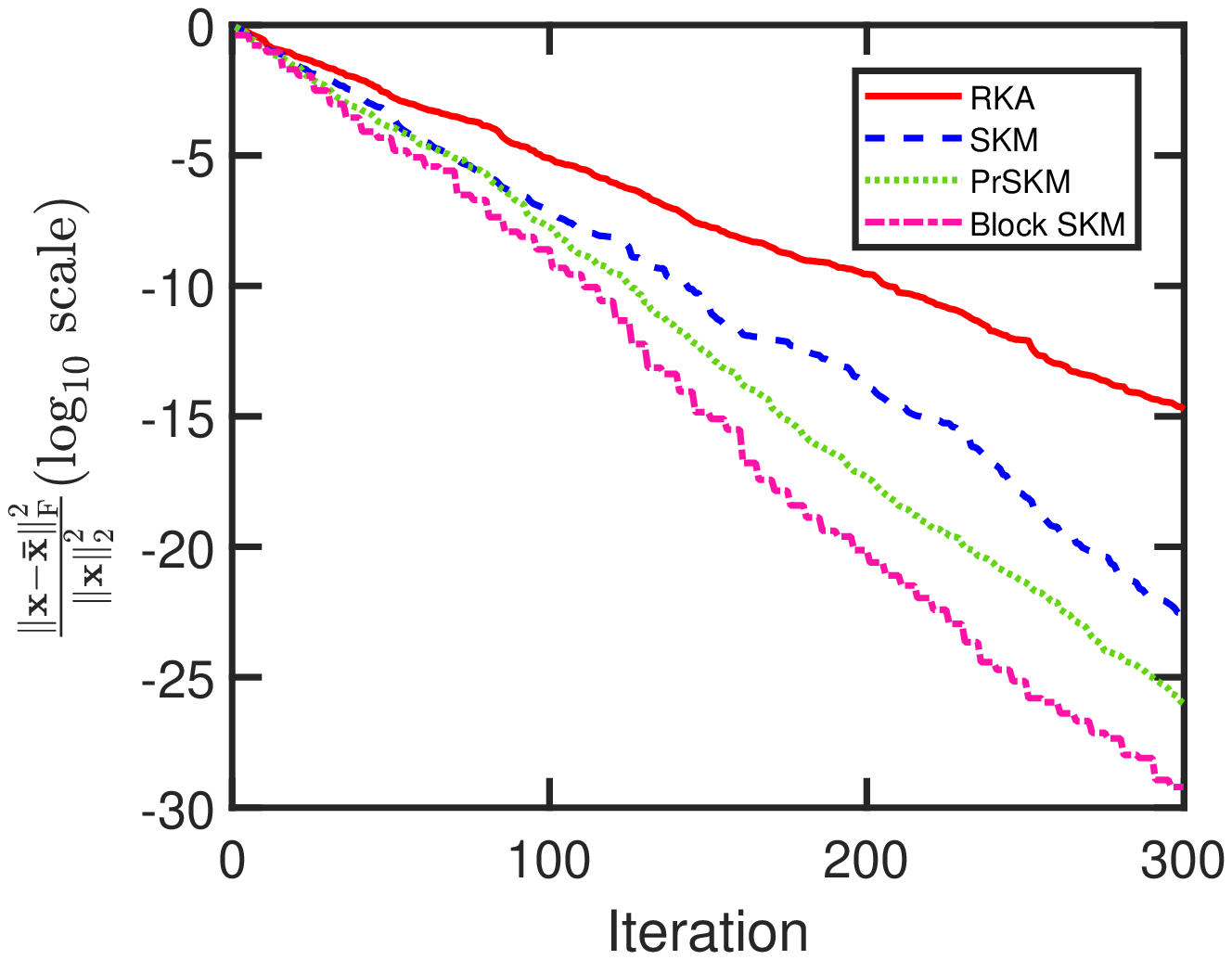}
		\caption{}
	\end{subfigure}
	\begin{subfigure}[b]{0.45\textwidth}
		\includegraphics[width=1\linewidth]{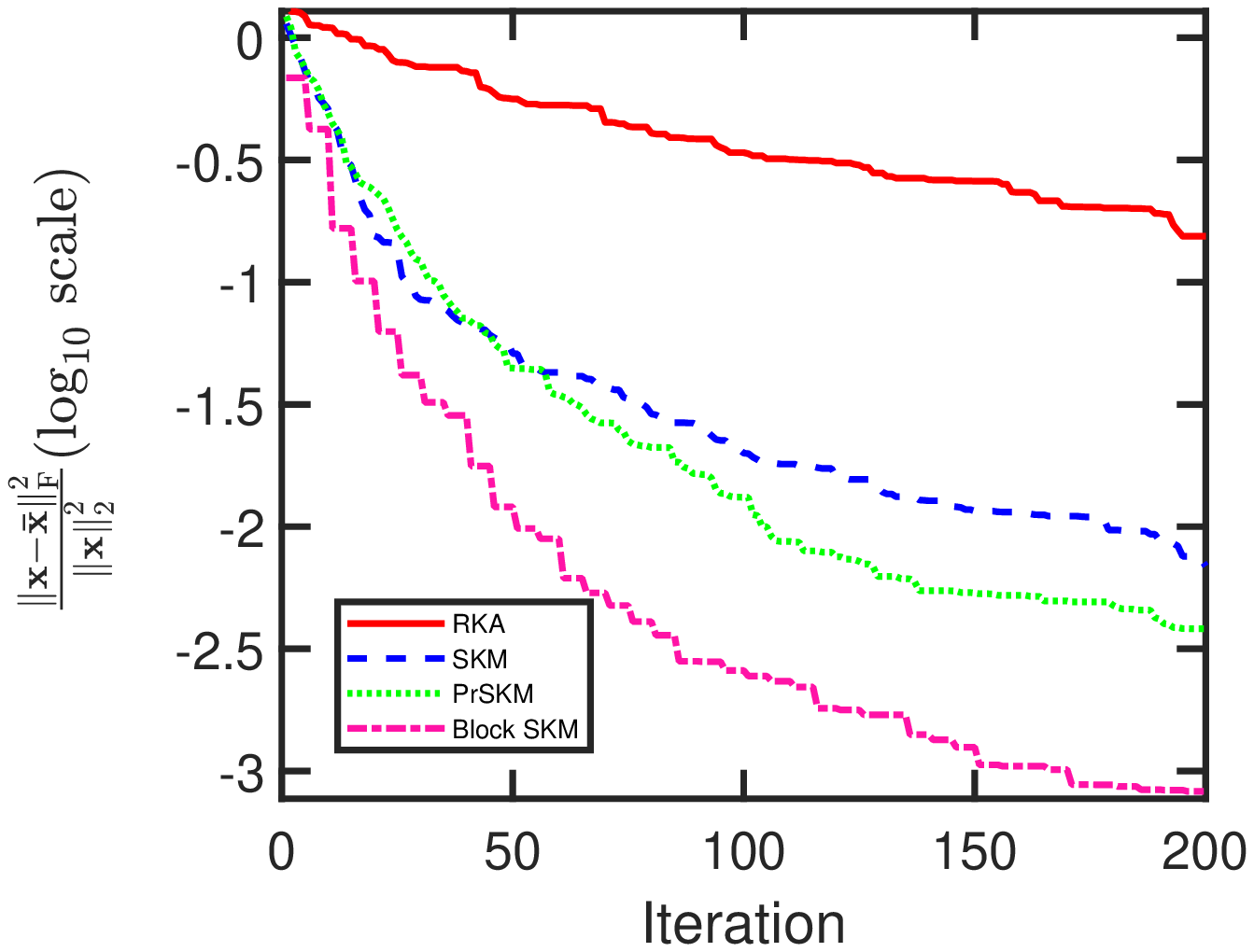}
		\caption{}
	\end{subfigure}
	\caption{Comparing the NMSE recovery performance of the two proposed Kaczmarz algorithms, namely the PrSKM and the block SKM, with that of SKM and RKA for: (a) a linear equation system, (b) a linear inequality system. }
	\label{figure_1}
\end{figure*}
\subsection{Comparing RKA, SKM, PrSKM and Block SKM}
\label{NUM_COM}
In this section, we numerically compare the RKA, SKM, PrSKM, and Block SKM in linear systems of equalities as well as those formed by inequalities.\\

\textbf{Linear feasibility of equalities}: Herein, we consider a block linear system of equalities $\mbA\mathbf{x}=\mathbf{b}$, where $\mbA=\left[\begin{array}{c|c|c}
\mbA_{1}^{\top} &\cdots &\mbA_{100}^{\top}\end{array}\right]^{\top}$, $\mbA_{i}\in\mathbb{R}^{10\times 10}$, $\mathbf{x}\in\mathbb{R}^{10}$, and $\mathbf{b}\in\mathbb{R}^{1000}$. Each row of $\mbA_{i}$ is generated as $\mathbf{a}_{j}^{i}\sim \mathcal{N}\left(\mathbf{0},\mbI_{10}\right)$. Also, the unknown signal $\mathbf{x}$ is generated as $\mathbf{x}\sim\mathcal{N}\left(\mathbf{0},\mbI_{10}\right)$.
The normalized mean square error (NMSE) is defined as
\begin{equation}
\label{eq:1700000}
\operatorname{NMSE} \triangleq \frac{\left\|\mathbf{x}^{\star}-\bar{\mathbf{x}}\right\|_{2}^{2}}{\left\|\mathbf{x}^{\star}\right\|_{2}^{2}},
\end{equation}
where $\mathbf{x}^{\star}$ and $\bar{\mathbf{x}}$ denote the true discretized signal and its recovered version, respectively.

Fig.~\ref{figure_1} illustrates the performance of RKA, SKM, PrSKM, and Block SKM in the recovery of $\mathbf{x}$ from the system $\mbA\mathbf{x}=\mathbf{b}$ with NMSE results. As can be observed, the Block SKM outperforms the other three approaches in the recovery task. Also, it can be seen that the PrSKM has a better recovery performance compared to that of the RKA and the SKM.\\

\textbf{Linear feasibility of inequalities}: We utilize ORKA to make a linear equation $\mbB\mathbf{x}=\mathbf{y}$ linear inequalities system, where the number of time-varying sampling threshold sequences is $m=40$, $\mbB\in\mathbb{R}^{100\times 10}$, $\mathbf{x}\in\mathbb{R}^{10}$, and $\mathbf{y}\in\mathbb{R}^{100}$. Each row of $\mbB$ is generated as $\mathbf{b}_{j}\sim\mathcal{N}\left(\mathbf{0},\mbI_{10}\right)$. Also, the desired signal $\mathbf{x}$ is generated as $\mathbf{x}\in\sim\mathcal{N}\left(\mathbf{0},\mbI_{10}\right)$. Each time-varying sampling threshold sequence $\boldsymbol{\uptau}^{(\ell)}$ is considered to have the distribution $\boldsymbol{\uptau}^{(\ell)}\sim\mathcal{N}\left(\mathbf{0},\mbI_{10}\right)$. The performance of the RKA, SKM, PrSKM, and Block SKM is illustrated in Fig.~\ref{figure_1}. Similar to the linear feasibility of equalities, it can be seen that the Block SKM has a better accuracy in the recovery of the desired signal $\mathbf{x}$ in the one-bit polyhedron (\ref{eq:80n}) compared to the other three approaches. The NMSE results in Fig.~\ref{figure_1} are averaged over $15$ experiments.

\section{Probabilistic Effect of Sample Abundance In ORKA}
\label{ada_thresh_prob}
An integral part of our proposed recovery algorithm is RKA, whose recovery error was readily given in (\ref{eq:150}). As shown in \cite{strohmer2009randomized}, the convergence rate of RKA does not depend on the number of equations in the system. We will show that the convergence rate of ORKA for linear feasibility is the same as RKA. Nevertheless, when we face a non-linear constraint in our problem, as is generally the case in (\ref{eq:1nnnn}), it is desirable are made redundant by using the opportunity of having a large number of samples; as typically provided via one-bit sampling. In such a case, The offered convergence rate appears to be insufficient since we must have enough number of samples to fulfill costly constraints. So, an extra term as a \emph{penalty} must be considered to present the importance of sample size in our algorithm \cite{9896984}.    

By adding more inequality constraints in (\ref{eq:80n}) as a result of extra one-bit samples, the shrinkage of the said polyhedron will put a downward pressure on the distance between the desired signal $\mathbf{x}^{\star}$ and its surrounding hyperplanes, each presenting an informative measurement that will shrink the feasibility space. We will show that by judicious sampling, the average of these distances will be bounded, which may be considered to be a finite-volume space created around $\mathbf{x}^{\star}$. Moreover, as a result of using an overdetermined linear system of inequalities, the convergence of the RKA is guaranteed \cite{9896984,leventhal2010randomized,needell2014paved,de2017sampling}.

\subsection{Recovery Error Upper Bound for ORKA}
\label{error}
As the scaled condition number is the central parameter governing the recovery error of Kaczmarz algorithms, we will evaluate it for ORKA-created matrix $\mbP$ in the following, starting with $\sigma_{min}$. The singular values of $\mbP$ may be determined based on the following theorem, which thus unveils the value of $\sigma_{min}$.
\begin{theorem}
\label{theorem_2}
In ORKA, the rank of $\mbP$ is equal to that of the constraint matrix $\mbA$, and its singular values are given by 
\begin{equation}
\label{singular}
\left\{\sigma_{i}\right\}=\sqrt{m}\left\{\sigma_{i\mbA}\right\},
\end{equation}
where $\left\{\sigma_{i\mbA}\right\}$ are singular values of $\mbA$, and $m$ is the number of time-varying sampling threshold sequences. Moreover, the scaled condition number of the ORKA-created matrix $\mbP$ is equal to that of the constraint matrix $\mbA$:
\begin{equation}
 \kappa\left(\mbP\right)=\kappa\left(\mbA\right).   
\end{equation}
\end{theorem}
\begin{proof}
To obtain the singular values of $\mbP$, the matrix $\mbW = \mbP^{\top}\mbP$ is computed as
\begin{equation}
\label{proof_th}
\begin{aligned}
\mbW &= \mbP^{\top}\mbP,\\&= \left[\mbA^{\top}\bOmega^{(1)}\vdots \mbA^{\top}\bOmega^{(2)}\vdots\cdots\vdots\mbA^{\top}\bOmega^{(m)}\right]
\left[\begin{array}{c}
\bOmega^{(1)}\mbA\\
\cdots\\
\bOmega^{(2)}\mbA\\
\vdots\\
\cdots\\
\bOmega^{(m)}\mbA
\end{array}\right],\\
&=\mbA^{\top}\bOmega^{(1)}\bOmega^{(1)}\mbA+\cdots+\mbA^{\top}\bOmega^{(m)}\bOmega^{(m)}\mbA,\\
&= m\mbA^{\top}\mbI\mbA=m\mbA^{\top}\mbA,
\end{aligned}
\end{equation}
which means the singular values of $\mbP$ are $\left\{\sigma_{i}\right\}=\sqrt{m}\left\{\sigma_{i\mbA}\right\}$. 

Also, the Frobenius norm of  
$\mbP$ is obtained as
\begin{equation}
\label{frob}
\begin{aligned}
\|\mbP\|^{2}_{\mathrm{F}}&=\operatorname{Tr}\left(\mbP^{\top}\mbP\right),\\ &= \operatorname{Tr}\left(m\mbA^{\top}\mbA\right) = m\|\mbA\|^{2}_{\mathrm{F}}.
\end{aligned}
\end{equation}
Consequently, the scaled condition number is independent of the number of time-varying thresholds sequences. It follows that $\kappa\left(\mbP\right)=\kappa\left(\mbA\right)$. 
\end{proof}
\begin{corollary}
\label{col_1}
For $\mbA=\mbI$, corresponding to the one-bit sampled signal recovery problem, the scaled condition number of $\bOmega=\left[\bOmega^{(1)}\vdots \bOmega^{(2)}\vdots\cdots\vdots\bOmega^{(m)}\right]$ is $\kappa\left(\Omega\right)=\sqrt{n}$, which is the infimum of the scaled condition number as shown in Theorem~\ref{scaled_number}.
\end{corollary}
Note that the convergence bound (\ref{eq:150}) for ORKA is independent of the number of time-varying sampling threshold sequences $m$, and it cannot take into account the effect of an increasing number of time-varying threshold sequences. Inspired by \cite{9896984}, we augment (\ref{eq:150}) with a sample size-dependent penalty function to make it useful in a sample abundance scenario:
\begin{proposition}[Convergence rate of ORKA]
\label{penaltyyy}
In the proposed recovery approach, it is deemed necessary to have a sufficient number of samples (inequalities) in order to guarantee a finite-volume feasible region and a bounded recovery error. Once our search area is located inside $\mathcal{F}_{\mbX}$, we may effectively deploy (\ref{eq:150}) for the convergence rate. The convergence rate of the Kaczmarz variants is useful when we have a linear feasibility problem. On the other hand, in (\ref{eq:1nnnn}), the main constraints are non-linear and they may be considered to be redundant by deploying enough samples \cite{9896984}. A sample size-aware convergence rate for ORKA may be formulated as:
\begin{equation}
\label{bound2}
\begin{aligned}
\mathbb{E}\left\{\hbar\left(\mathbf{x}_{i},\mathbf{x}^{\star}\right)\right\} \leq \left(1-\frac{2\lambda_{i}-\lambda^{2}_{i}}{\kappa^{2}\left(\mbA\right)}\right)^{i} \hbar\left(\mathbf{x}_{0},\mathbf{x}^{\star}\right)+\Upsilon\left(m\right),
\end{aligned}
\end{equation}
where $	\Upsilon(.)$ is an asymptotically decreasing function, such that if the number of time-varying threshold sequences is enough for the one-bit polyhedron to fit inside $\mathcal{F}_{\mbX}$, the sample size-dependent penalty function $\Upsilon\left(m\right)$ approaches zero.
\end{proposition}
To propose an appropriate sample size-dependent penalty function, we will utilize the first theorem in \cite{9896984}, which studies the possibility of creating a finite-volume space around the optimal signal.
\subsection{Sample Size-Dependent Penalty Function via Moment Generating Functions}
\label{penalt}
We investigate the convergence of ORKA through a probabilistic lens. To do so, define the distance between the optimal point $\mathbf{x}^{\star}$ and the $j$-th hyperplane presented in (\ref{eq:80n}) as
\begin{equation}
\label{distance}
\begin{aligned}
d_{j}\left(\mathbf{x}^{\star},\boldsymbol{\uptau}^{(\ell)}\right) &= \left|r_{j}\odot\left(\mba_{j}\mathbf{x}^{\star}-\tau^{(\ell)}_{j}\right)\right|^{2},\quad j\in\left\{1,\cdots,m^{\prime}\right\},
\end{aligned}
\end{equation}
where $r_{j}\odot\mba_{j}$ is the $j$-th row of $\mbP$, $\mba_{j}$ is the $j$-th row of $\mbA$, and $\mba_{j}=\mba_{j+n}$. It is easy to observe that by generally reducing the distances between $\mathbf{x}^{\star}$ and the constraint-associated hyperplanes, the possibility of \emph{capturing} the optimal point is increased. For a specific sample size $m^{\prime}=m n$, when the volume of the finite space around the optimal point is reduced, the average of $\left\{d_{j}\left(\Tilde{\mathbf{x}}^{\star},\boldsymbol{\uptau}^{(\ell)}\right)\right\}$ is diminished as well. This average of distance can be written as \cite{leventhal2010randomized}:
\begin{equation}
\label{ave}
\begin{aligned}
T_{\text{ave}} =  \frac{1}{m^{\prime}}\sum^{m^{\prime}}_{j=1}d_{j}\left(\Tilde{\mathbf{x}}^{\star},\boldsymbol{\uptau}^{(\ell)}\right).
\end{aligned}
\end{equation}
The possibility of creating a finite-volume, and the importance of the number of samples in the recovery performance of RKA, can be captured by the \emph{Chernoff bound} as illustrated below.
\begin{theorem}[See\cite{9896984}]
\label{theorem_0}
Assume the distances $\left\{d_{j}\left(\Tilde{\mathbf{x}}^{\star},\boldsymbol{\uptau}^{(\ell)}\right)\right\}$ between the desired point $\Tilde{\mathbf{x}}^{\star}$ and the hyperplanes of the polyhedron defined in (\ref{eq:24}) are i.i.d. random variables. Then:
\begin{itemize}
    \item The Chernoff bound of $T_{\text{ave}}$ is given by
\begin{equation}
\label{eq:theorem_cher}
\operatorname{Pr}\left(\frac{1}{m^{\prime}}\sum^{m^{\prime}}_{j=1}d_{j}\left(\Tilde{\mathbf{x}}^{\star},\boldsymbol{\uptau}^{(\ell)}\right)\leq a\right)\geq 1-\inf_{t\geq 0}\frac{\Psi_{T}}{e^{t a}},
\end{equation}
where $a$ is an average distance point in space at which the finite-volume space around the desired signal is created, and $\Psi_{T}$ is the moment generating function (MGF) of the error recovery, given as
\begin{equation}
\label{eq:psi}
\Psi_{T} = \left(1+t\frac{\mu^{(1)}_{d_{j}}}{m^{\prime}}+\cdots+t^{\kappa}\frac{\mu^{(\kappa)}_{d_{j}}}{\kappa!m^{\prime\kappa}}+\mathcal{R}\left(m^{\prime}\right)\right)^{m^{\prime}},
\end{equation}
with $\mu^{(\kappa)}_{d_{j}}=\mathbb{E}\left\{d^{\kappa}_{j}\right\}$, and $\mathcal{R}$ denoting a bounded remainder associated with truncating the Taylor series expansion of $\Psi_{T}$.
\item $\Psi_{T}$ is decreasing with an increasing sample size in the sample abundance scenario, leading to an increasing lower bound in (\ref{eq:theorem_cher}).
\end{itemize}
\end{theorem}
The MGF $\Psi_{T}$ contains two parts. The first part has an increasing trend until a specific sample size $m^{\star}$, which indicates the existence of an abundant number of samples. After that, the function has a decreasing behavior. Therefore, $\Psi_{T}$ with $m\geq m^{\star}$ can be a good choice for a sample size-dependent penalty function. Particularly, the penalty function can be chosen as $\Psi_{T}-\Psi_{\infty}$, where $\Psi_{\infty}=\lim_{m\rightarrow \infty} \Psi_{T}$, to ensure $\Upsilon(m)\rightarrow 0$ as $m\rightarrow \infty$. 

Since, we do not have access to the probability density function of $\left\{d_{j}\right\}$, thus, the MGF must be evaluated by the truncated Taylor series expansion, which may be accurately approximated by a rational function such as a \emph{Padé approximation} (PA). The decreasing part of $\Psi_{T}$ in $m> m^{\star}$ with PA is modeled as follows\cite{eamaz2021modified,eamaz2022covariance}:
\begin{equation}
\label{pade_1}
\begin{aligned}
\Upsilon(m)&\asymp \left(1+\cdots+t^{\kappa}\frac{\mu^{(\kappa)}_{d_{j}}}{\kappa!m^{\prime\kappa}}\right)^{m^{\prime}}-\Psi_{T},\\
&=\frac{a_{0}+\frac{a_{1}}{m}}{b_{0}+\frac{b_{1}}{m}}-\frac{a_{0}}{b_{0}},
\end{aligned}
\end{equation}
where $\left\{a_{0}, a_{1}, b_{0}, b_{1}\right\}$ are the PA coefficients as given by
\begin{equation}
\label{eq:coeff}
\begin{aligned}
a_{0}&=e^{u}\left(12u^{2}-24v\right),\\a_{1}&=e^{u}\left(-3u^{4}+8u^{3}+12u^{2}v-24uv-12v^{2}\right),\\b_{0}&=12u^{2}-24v,\\b_{1}&=3u^{4}+8u^{3}-12u^{2}v-24uv+12v^{2},
\end{aligned}
\end{equation}
where $u=\mu_{d_{j}}^{(1)}t$ and $v=\frac{\mu_{d_{j}}^{(2)}t^{2}}{2}$.

\section{Judicious Sampling With Adaptive Thresholding for ORKA}
\label{ada_thresh}
By the spirit of using the iterative RKA, a suitable time-varying threshold can be selected in order to enhance the recovery performance. In ORKA, we face a highly overdetermined linear feasibility problem creating a finite-volume space. To capture the desired signal $\mathbf{x}^{\star}$ more efficiently, the right-hand side of the inequalities in (\ref{eq:80n}), i.e. $\operatorname{vec}\left(\mbR\right)\odot \operatorname{vec}\left(\bGamma\right)$, must be determined in a way that each associated hyperplane passes through the desired feasible region within $\mathcal{F}_{\mbX}$. Therefore, an algorithm is proposed to ensure that this occurs in practice.

We propose an iterative algorithm generating an adaptive sampling threshold to accurately obtain the desired solution. To have the smaller area of the finite-volume space around the desired signal $\mathbf{x}^{\star}$, one can somehow choose thresholds to reduce distances between them and the desired point. To do so, we update the time-varying threshold for $\ell\in\left\{1,\cdots,m\right\}$ as
\begin{equation}
\label{eq:200}
\mbA\mathbf{x}_{k}-\mbr^{(\ell)}_{k}\odot\bepsilon^{(\ell)}_{k}= \boldsymbol{\uptau}^{(\ell)}_{k+1},
\end{equation}
where $\bepsilon^{(\ell)}_{k}$ are positive vectors in the $k$-th iteration of the algorithm. This updating process is based on
\begin{equation}
\label{akhund}
r_{j}=\begin{cases} +1 &\mba_{j}\mathbf{x}>\uptau_{j}, \\ -1&\mba_{j}\mathbf{x}<\uptau_{j},\end{cases}
\end{equation}
where $\{\mba_{j}\}$ are the rows of $\mbA$.
The one-bit measurements $\left\{\mbr^{(\ell)}_{k}\right\}$ are updated in the way to satisfy (\ref{eq:80n}) in iteration $k$, i.e. the inequalities $\Omega^{(\ell)}_{k} \mbA\mathbf{x}^{\star}\geq \mbr^{(\ell)}_{k}\odot\boldsymbol{\uptau}^{(\ell)}_{k}$. The reason behind this updating is to ensure that each halfspace associated with a threshold in iteration $k$ is getting closer to the optimal point in the correct direction which means the main side of the halfspace is forced to cover the optimal solution.    

\begin{proposition}[ORKA with Adaptive Thresholding]
\label{adaptive}
Consider applying ORKA to a linear feasibility problem $\mbA\mathbf{x}= \mathbf{y}$ as part of the linear constraints of (\ref{eq:1nnnn}). Suppose the initial time-varying threshold sequences are $\left\{\boldsymbol{\uptau}^{(\ell)}_{0}\right\}\sim \mathcal{N}\left(0,1\right)$ (with the same length as $\mbr^{(\ell)}$), and $\left\{\delta^{(\ell)}\right\}$ are positive numbers. Also, $\mathbf{x}_{k}$, $\boldsymbol{\uptau}^{(\ell)}_{k}$, $\mbr^{(\ell)}_{k}$ and $\bepsilon_{k}$ denote their associated values at iteration $k$. The proposed sampling algorithm is summarized as follows:
\begin{enumerate}
    \item Find a point inside the following polyhedron with proposed accelerated Kaczmarz algorithms, i.e. the PrSKM or the block SKM for $\boldsymbol{\uptau}=\left\{\boldsymbol{\uptau}^{(\ell)}_{k}\right\}$:
    \begin{equation}
    \mathcal{P}_{k}=\left\{\mathbf{x}_{k} \mid \Tilde{\bOmega}_{k}\mbA\mathbf{x}_{k}\succeq\mbb_{k}\right\},  
    \end{equation}
    where $\mbb_{k}=\operatorname{vec}\left(\mbR_{k}\right)\odot \operatorname{vec}\left(\bGamma_{k}\right)$,  $\mbR_{k}$ and $\bGamma_{k}$ are matrices with $\left\{\mbr^{(\ell)}_{k}\right\}$ and $\left\{\boldsymbol{\uptau}^{(\ell)}_{k}\right\}$ representing their columns, respectively.
    \item Update $\bGamma_{k+1}$ as:
    \begin{equation}
    \operatorname{vec}\left(\mbR_{k}\right)\odot \operatorname{vec}\left(\bGamma_{k+1}\right)=\Tilde{\bOmega}_{k}\mbA\mathbf{x}_{k}-\frac{\bepsilon_{k}}{2}.
    \end{equation}
    \item Compute $\bepsilon_{k}$, a block vector containing $\left\{\bepsilon^{(\ell)}_{k}\right\}$, as:
    \begin{equation}
    \bepsilon_{k}=\Tilde{\bOmega}_{k}\mbA\mathbf{x}_{k}-\operatorname{vec}\left(\mbR_{k}\right)\odot \operatorname{vec}\left(\bGamma_{k}\right).    
    \end{equation}
    \item Update $\mbR_{k+1}$ based on (\ref{akhund}):
    \begin{equation}
    \mbr^{(\ell)}_{k+1}=\operatorname{sgn}\left(\mathbf{y}-\boldsymbol{\uptau}^{(\ell)}_{k+1}\right),\quad \ell \in \left\{1,\cdots,m\right\}.
    \end{equation}
    \item Increase $k$ by one.
    \item Stop when $\left\|\boldsymbol{\uptau}^{(\ell)}_{k+1}-\boldsymbol{\uptau}^{(\ell)}_{k}\right\|_{2}\leq \delta^{(\ell)}$.
\end{enumerate}
\end{proposition}
One can observe that by deploying this adaptive thresholding algorithm, smaller values of $\{d_{j}\}$ will emerge which leads to their moments $\left\{\mu^{(\kappa)}_{d_{j}}\right\}$ to further diminish. Therefore, $\Psi_{T}$ is smaller in this scenario and a smaller number of time-varying sampling threshold sequences can be utilized in ORKA with similar recovery performance. Additionally, non-informative sampling thresholds, which appear as extra inequality constraints in the random time-varying sampling thresholds scenario, may be efficiently removed by choosing the adaptive thresholds with closer hyperplanes to the desired point.
\begin{figure*}[t]
	\centering
	\begin{subfigure}[b]{0.32\textwidth}
		\includegraphics[width=1\linewidth]{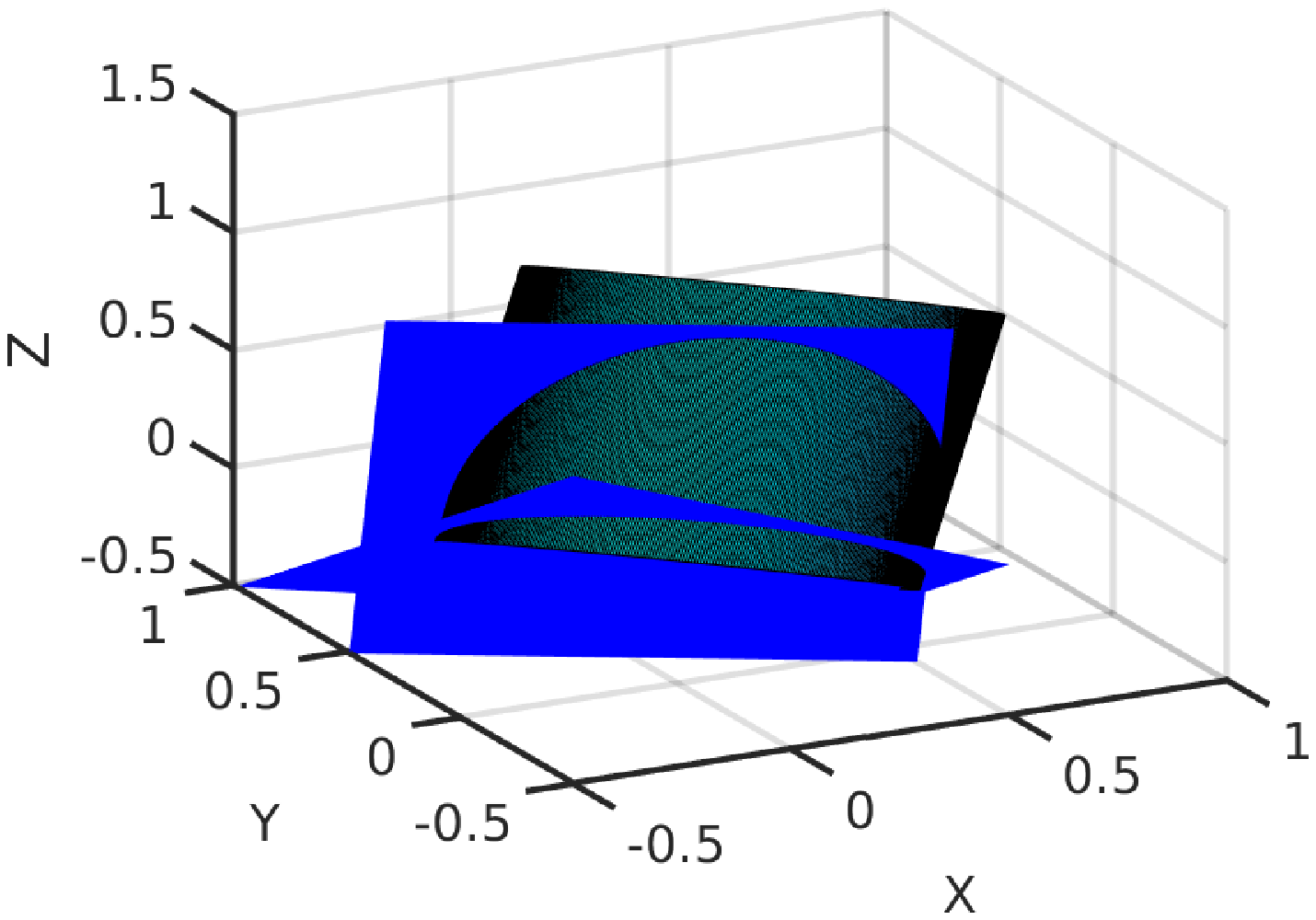}
		\caption{$m=2$}
	\end{subfigure}
	\begin{subfigure}[b]{0.32\textwidth}
		\includegraphics[width=1\linewidth]{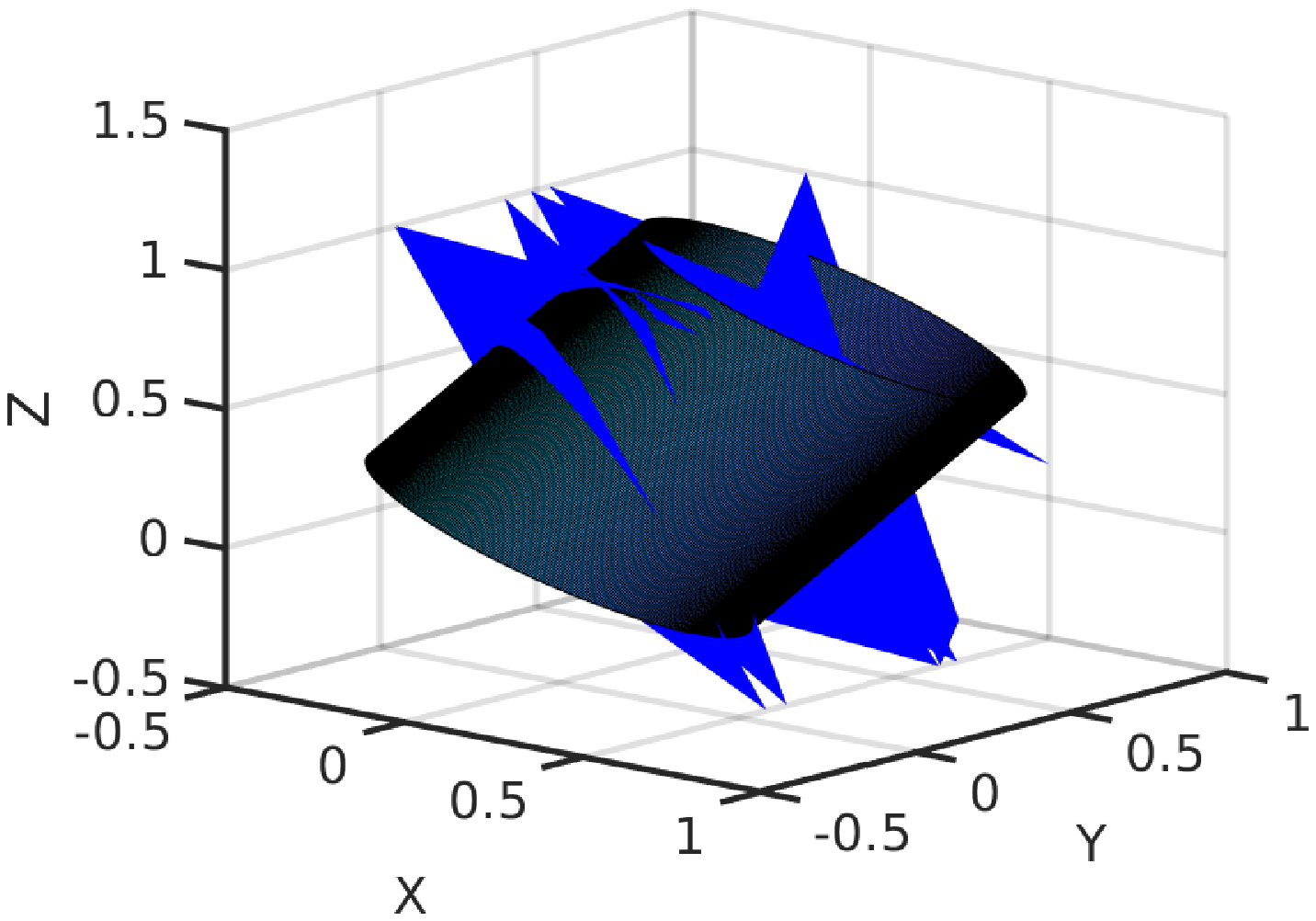}
		\caption{$m=6$}
	\end{subfigure}
	\begin{subfigure}[b]{0.32\textwidth}
		\includegraphics[width=1\linewidth]{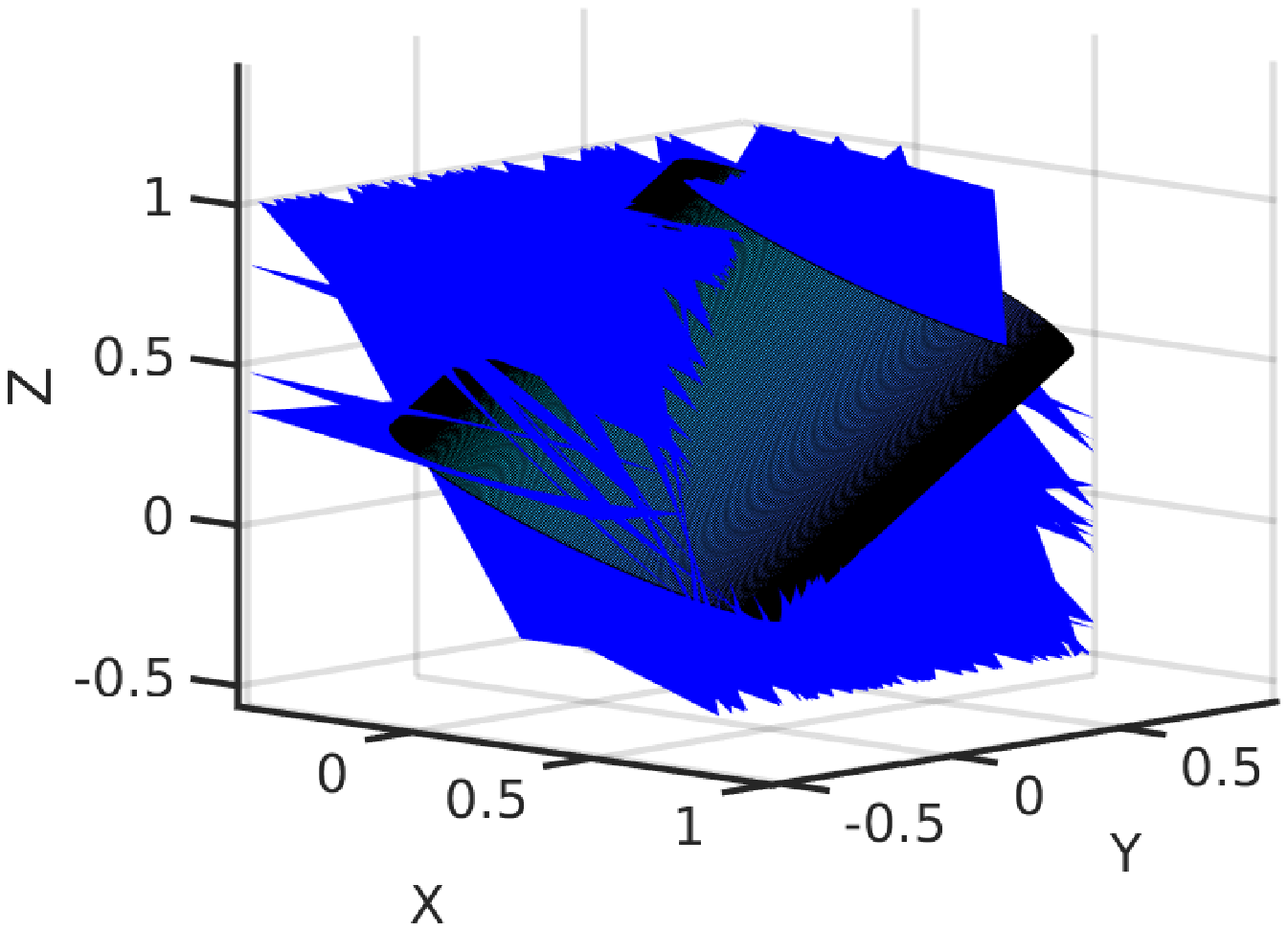}
		\caption{$m=60$}
	\end{subfigure}
		\begin{subfigure}[b]{0.32\textwidth}
		\includegraphics[width=0.8\linewidth]{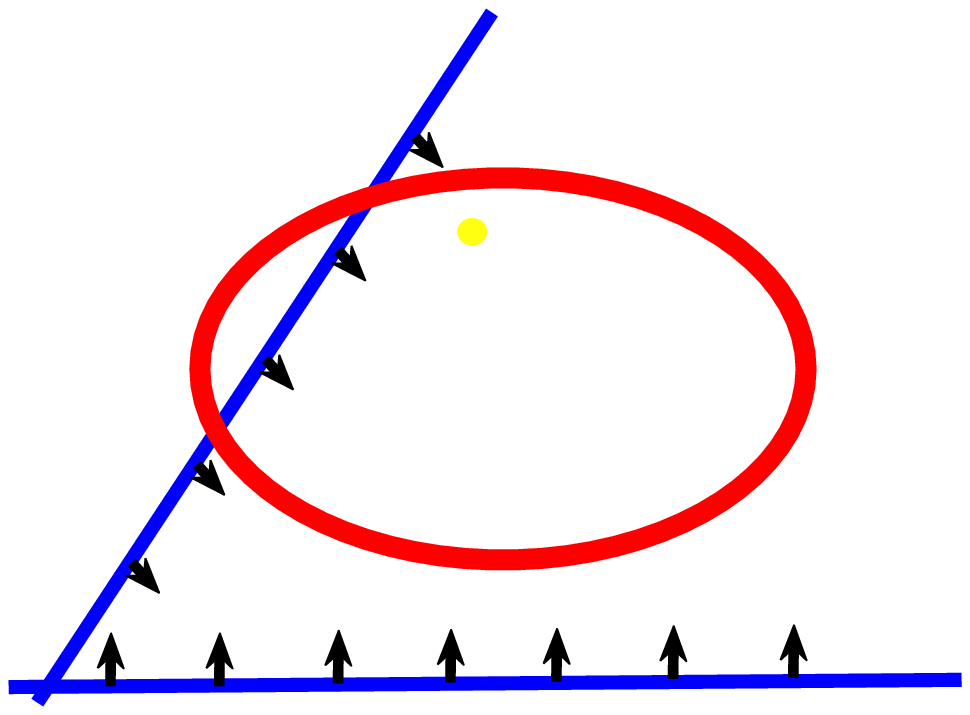}
		\caption{$m=2$}
	\end{subfigure}
		\begin{subfigure}[b]{0.32\textwidth}
		\includegraphics[width=0.8\linewidth]{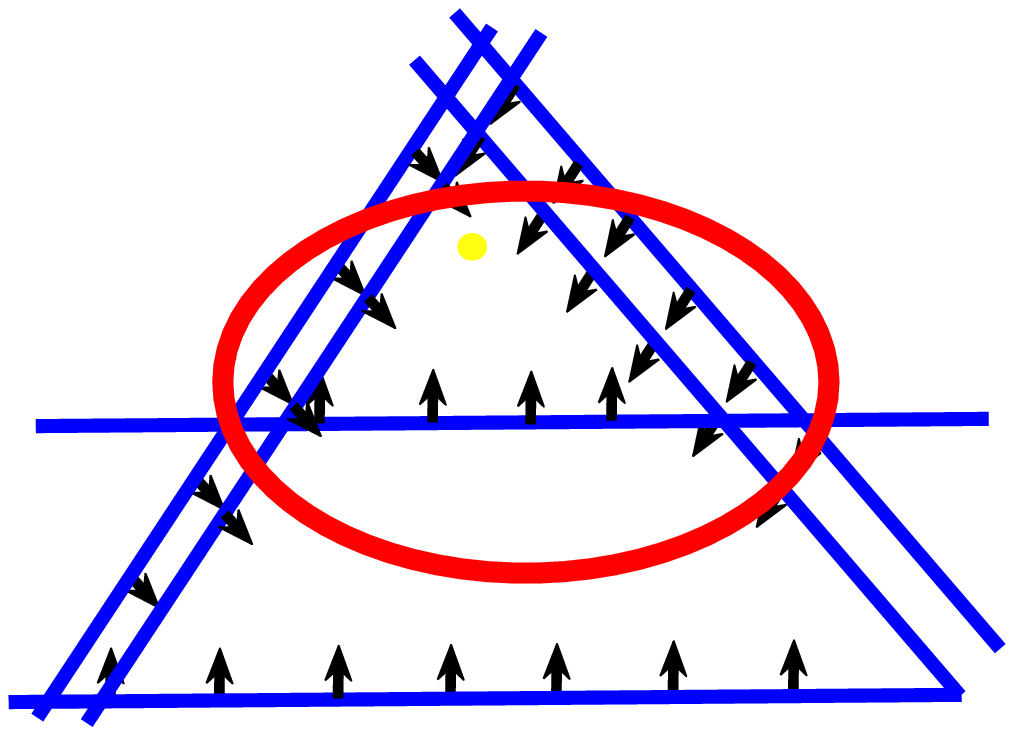}
		\caption{$m=6$}
	\end{subfigure}
		\begin{subfigure}[b]{0.32\textwidth}
		\includegraphics[width=0.8\linewidth]{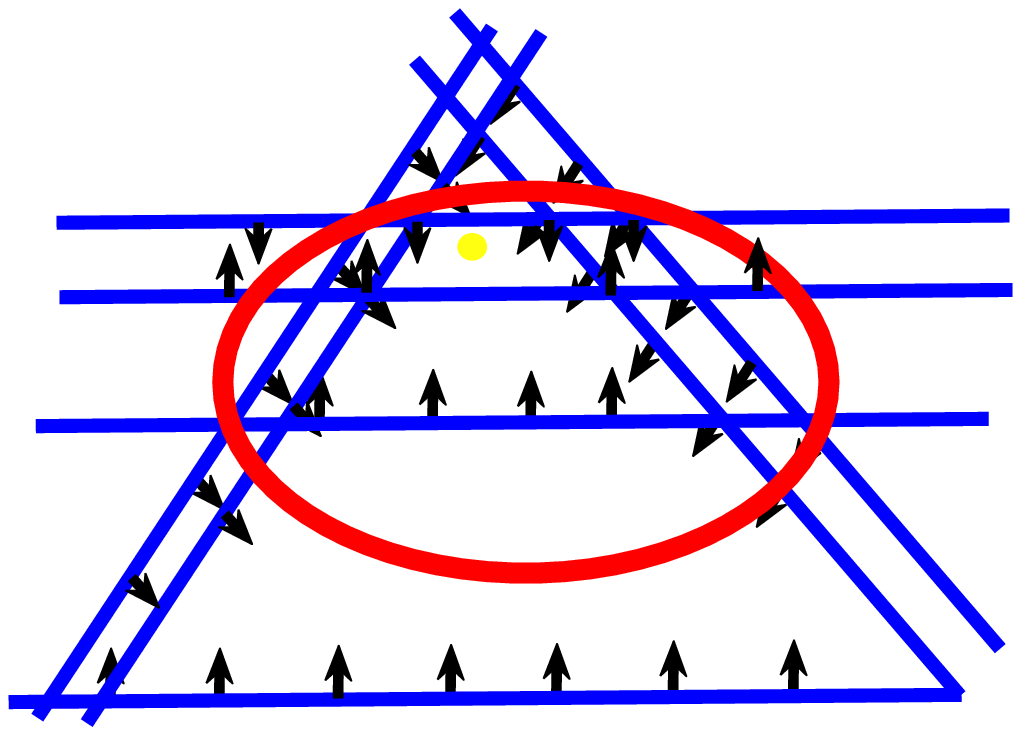}
		\caption{$m=60$}
	\end{subfigure}
	
	\caption{Shrinkage of the one-bit polyhedron (\ref{eq:80n}) in blue, ultimately placed within the unit ball of the nuclear norm $\left\|\mbX\right\|_{\star}\leq 1$ shown with black cylindrical region and its red contours, when the number constraints (samples) grows large. The arrows point to the half-space associated with each inequality constraint. The evolution of the feasible regime is depicted with increasing samples in three cases: (a) and (d) small sample-size regime, constraints not forming a finite-value polyhedron; (b) and (e) medium sample-size regime, constraints forming a finite-volume polyhedron, parts of which are outside the cylinder; (c) and (f) large sample-size regime, constraints forming a finite-volume polyhedron inside the nuclear norm cylinder, making its constraint redundant. The optimal point representing the signal to be recovered is shown by yellow.}
	
	\label{figure_1n}
\end{figure*}
\section{Low-Rank Matrix Recovery Via ORKA}
\label{MATRIX}
As mentioned earlier, low-rank matrix recovery is an excellent example for problems that assume the form in (\ref{eq:1nnnn}), and that can be tackled using our methodology. In this section, at first, we will briefly introduce the \emph{nuclear norm minimization} form of the problem. Accordingly, we will apply ORKA to this problem without considering the associated costly constraints. At the end, the recovery performance of ORKA will be numerically evaluated considering different matrix ranks and sample sizes to investigate the existence of a sample abundance scenario. 
\subsection{Problem Formulation}
\label{pr}
The problem of the low-rank matrix recovery can be formulated as:
\begin{equation}
\label{eq:1nnnnn}
\begin{aligned}
\text{find}\quad &\mbX\\
\text{s.t.} \quad &\mathcal{A}\left(\mbX\right)=\mathbf{y}, \\
& \operatorname{rank}\left(\mbX\right)\leq M,\\
&\mbX \in \Omega_{c},
\end{aligned}
\end{equation}
where $\mbX\in \mathbb{C}^{n_{1}\times n_{2}}$ is the matrix of unknowns, $\mathbf{y}\in \mathbb{R}^{n}$ is the measurement vector, and $\mathcal{A}$ is a linear transformation mapping $n_{1}\times n_{2}$ into $\mathbb{R}^{n}$. In general, $\Omega_{c}$ can be chosen such as the set of semi-definite matrices, symmetric matrices, upper or lower triangle matrices, Hessenberg matrices and a specific constraint on the matrix elements $\left\|\mbX\right\|_{\infty}\leq \alpha$ or on its eigenvalues, i.e., $\lambda_{i}\leq \epsilon$ where $\left\{\lambda_{i}\right\}$ are eigenvalues of $\mbX$ \cite{davenport2016overview,candes2015phase,van1996matrix}.

The problem (\ref{eq:1nnnnn}) can be rewritten as an optimization problem:
\begin{equation}
\label{eq:1nnnnnn}
\begin{aligned}
\min_{\mbX}\quad & \operatorname{rank}\left(\mbX\right)\\
\text{s.t.} \quad &\mathcal{A}\left(\mbX\right)=\mathbf{y}, \\
&\mbX \in \Omega_{c}.
\end{aligned}
\end{equation}
This problem is known to be NP-hard, whose solution is difficult to approximate \cite{meka2008rank,recht2011null}. Recall that the rank of $\mbX$ is equal to the number of nonzero singular values. In the
case when the singular values are all equal to one, the sum of the singular values is equal to the rank. When the singular values are less than or equal to one, the sum of
the singular values is a convex function that is strictly less than the rank. Therefore, it is been popular for this problem to replace the rank function with
the sum of the singular values of $\mbX$; i.e., its nuclear norm. The nuclear norm minimization alternative of the problem is given by \cite{cai2010singular,recht2010guaranteed,recht2011null}:
\begin{equation}
\label{eq:1nnnnnnn}
\begin{aligned}
\min_{\mbX}\quad & \left\|\mbX\right\|_{\star}\\
\text{s.t.} \quad &\mathcal{A}\left(\mbX\right)=\mathbf{y},\\
&\mbX \in \Omega_{c}.
\end{aligned}
\end{equation}
In this problem, the feasible set $\mathcal{F}_{\mbX}$ is obtained as
\begin{equation}
\label{gongen}
\mathcal{F}_{\mbX}=\left\{\mathcal{P}^{\star}\cap \Omega_{c}\right\},
\end{equation}
where $\mathcal{P}^{\star}$ is defined as follows
\begin{equation}
\label{stephanie}
\mathcal{P}^{\star} = \left\{\mbX \mid
\left\|\mbX\right\|_{\star}\leq \tau\right\},\quad
\tau\in\mathbb{R}^{+}.
\end{equation}
Next, we will apply ORKA to (\ref{eq:1nnnnnn}) to make its costly constraints redundant by using abundant number of time-varying sampling thresholds $m$.

A numerical investigation of (\ref{eq:80n}) when it is achieved for the nuclear norm minimization, reveals that by increasing the number of time-varying sampling threshold sequences $m$, the space formed by the intersection of half-spaces (inequality constraints) can fully shrink to the desired signal $\mbX^{\star}$ inside the feasible region of (\ref{stephanie}) which is shown by the cylindrical space \cite{recht2011null}---see Fig.~\ref{figure_1n} for an illustrative example of this phenomenon. As can be seen in this figure, the blue lines displaying the linear feasibility form a finite-volume space around the optimal point displayed by the yellow circle inside the cylinder (the elliptical region) by growing the number of threshold sequences or one-bit samples. In (a)/(d), constraints are not enough to create a finite-volume space, whereas in (b)/(e) such constraints can create the desired finite-volume polyhedron space which, however, is not fully inside the cylinder. Lastly, in (c)/(f), the created finite-volume space shrinks to be fully inside the cylinder.
\begin{figure*}[t]
	\centering
	\begin{subfigure}[b]{0.45\textwidth}
		\includegraphics[width=1\linewidth]{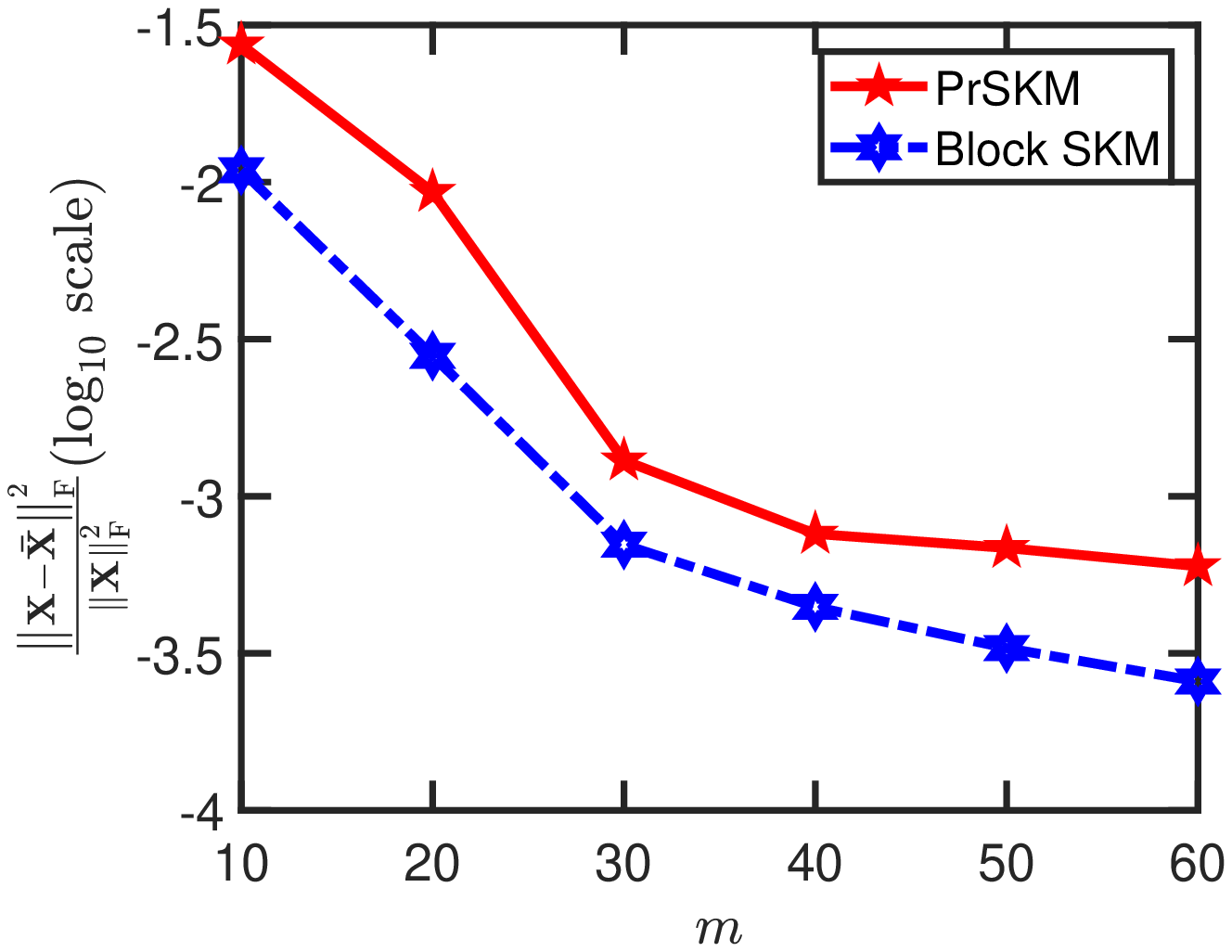}
		\caption{}
	\end{subfigure}
	\begin{subfigure}[b]{0.45\textwidth}
		\includegraphics[width=1\linewidth]{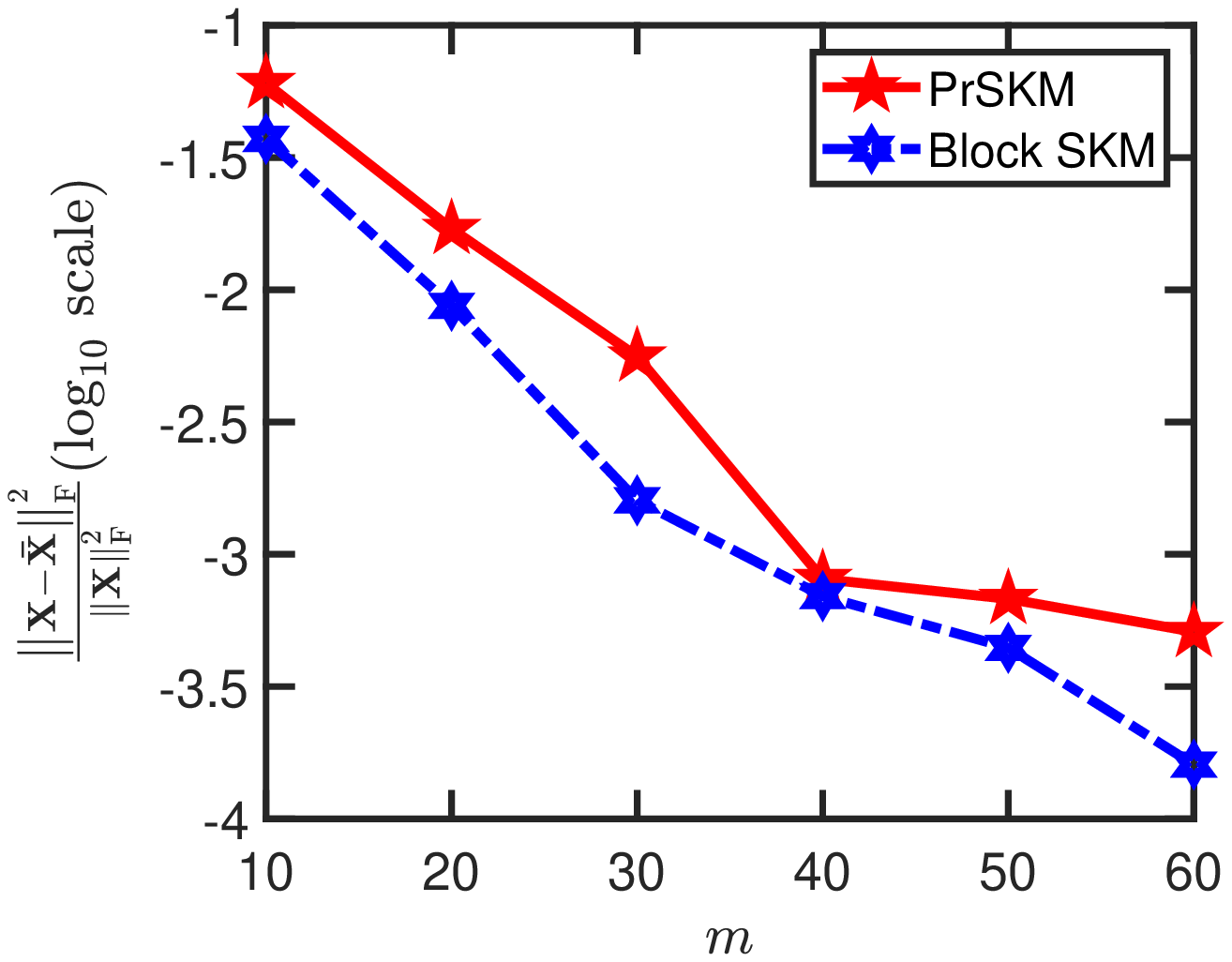}
		\caption{}
	\end{subfigure}
	\caption{Average NMSE for the Frobenius norm of error for the recovery of the matrix $\mbX$ associated with different time-varying sampling threshold sequences sizes when the PrSKM and the block SKM are utilized in ORKA: (a) $\operatorname{rank}\left(\mbX\right)=1$, (b) $\operatorname{rank}\left(\mbX\right)=4$.}
	\label{figure_2}
\end{figure*}
\begin{figure*}[t]
	\centering
	\begin{subfigure}[b]{0.45\textwidth}
		\includegraphics[width=1\linewidth]{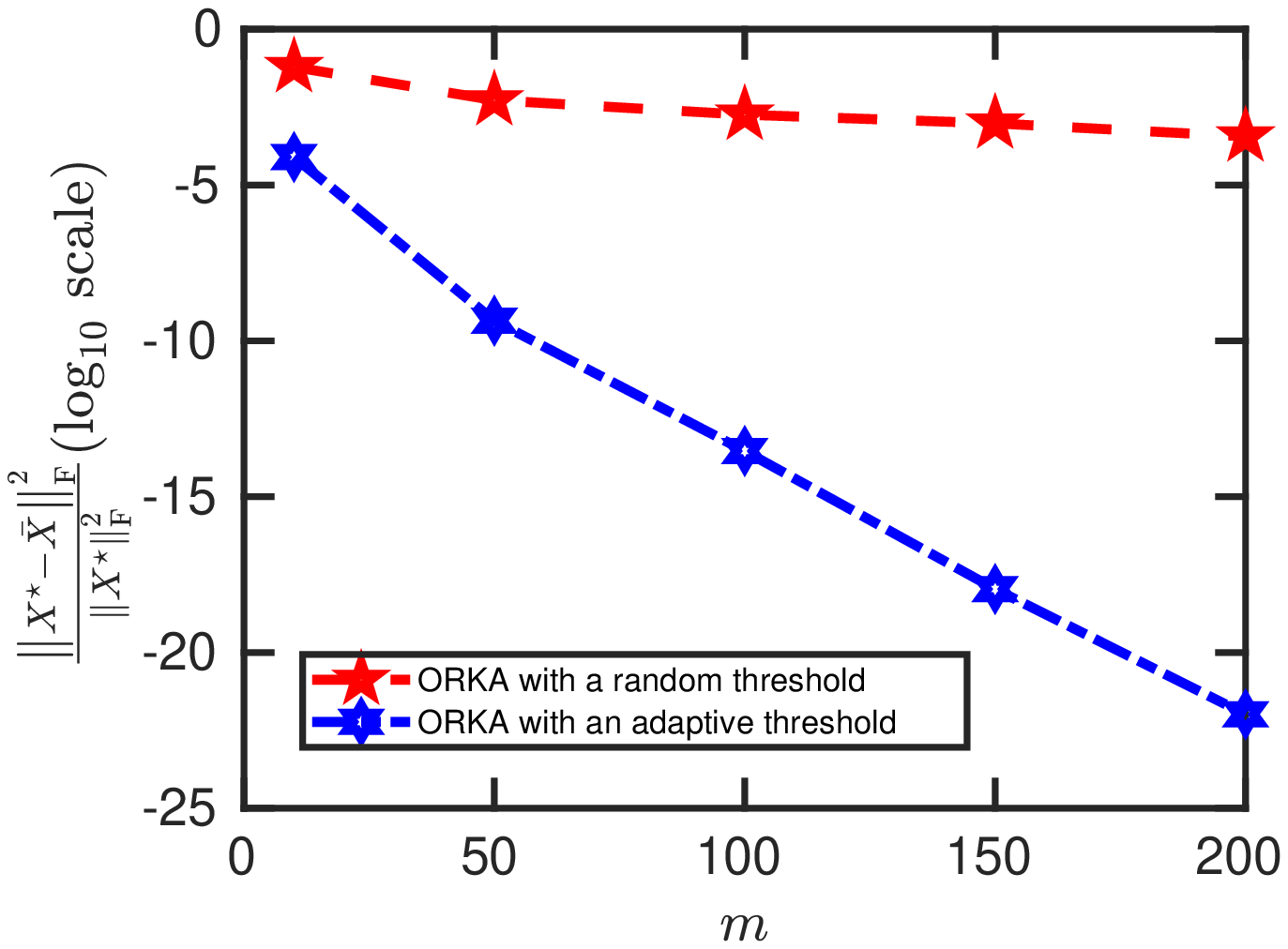}
		\caption{}
	\end{subfigure}
	\begin{subfigure}[b]{0.45\textwidth}
		\includegraphics[width=1\linewidth]{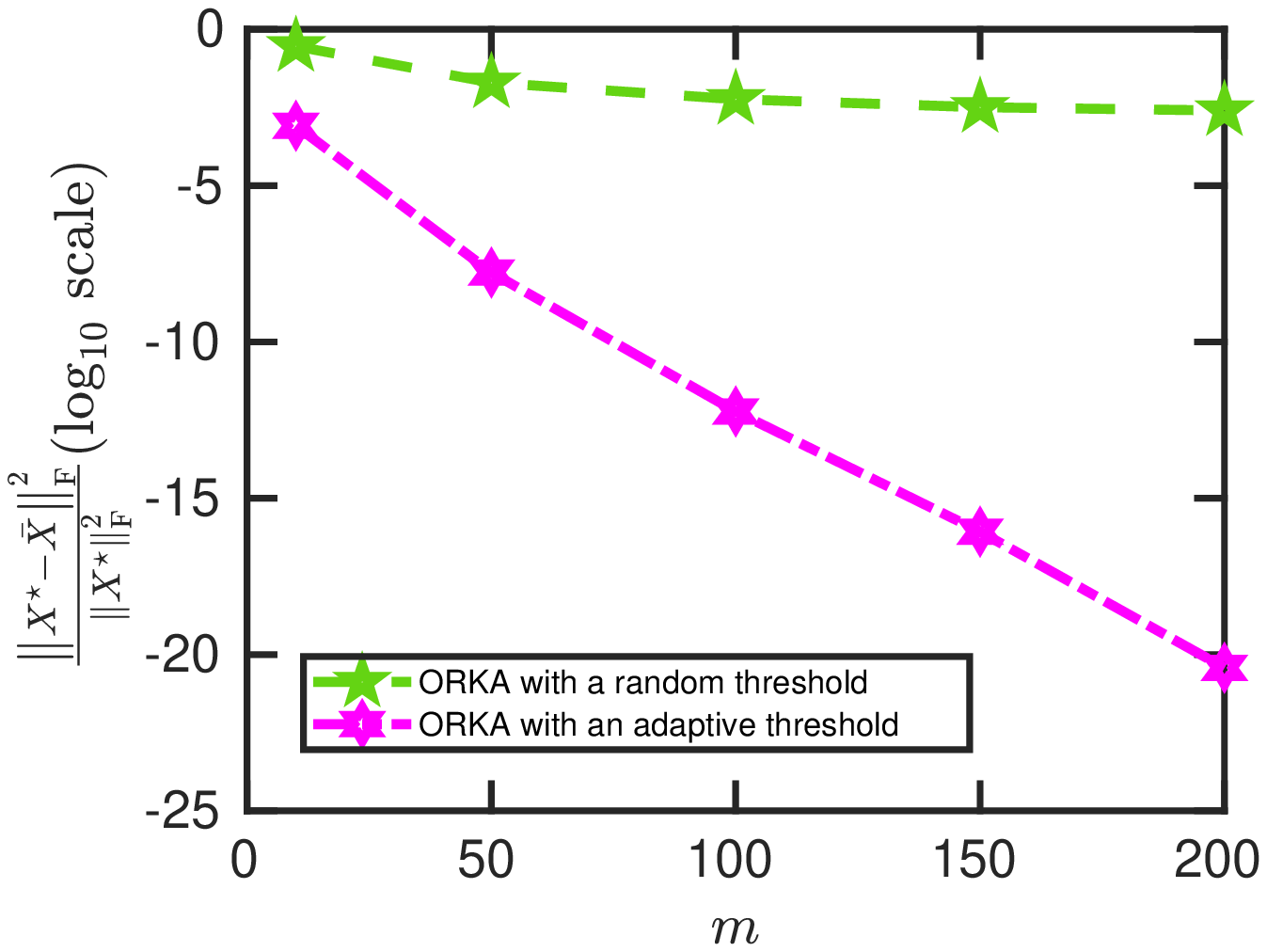}
		\caption{}
	\end{subfigure}
	\caption{Comparing the average NMSE for the Frobenius norm of error for the recovery of the matrix $\mbX$ using ORKA when (i) a random threshold and (ii) the adaptive sampling threshold are adopted when the PrSKM and the block SKM are utilized in ORKA: (a) $\operatorname{rank}\left(\mbX\right)=1$, (b) $\operatorname{rank}\left(\mbX\right)=4$.}
	\label{figure_3}
\end{figure*}
\subsection{Numerical Illustrations}
\label{NUM_matrix}
In this section, we numerically scrutinize the capability of the ORKA in the nuclear norm minimization problem (\ref{eq:1nnnnnnn}) instead of (\ref{eq:1nnnnnnn}) by the squared Frobenius norm of the error normalized by the squared Frobenius norm of the desired matrix $\mbX^{\star}$, defined as
\begin{equation}
\label{eq:4000}
\mathrm{NMSE}\triangleq\frac{\left\|\mbX^{\star}-\bar{\mbX}\right\|^{2}_{\mathrm{F}}}{\left\|\mbX^{\star}\right\|^{2}_{\mathrm{F}}}.
\end{equation}
We solve the overdetermined one-bit polyhedron in (\ref{eq:80n}) via the PrSKM and the Block SKM. To make this happen, we obtain the one-bit polyhedron from a linear feasibility problem $\mbA\mathbf{x}=\mathbf{y}$, where $\mbA\in\mathbb{R}^{200\times 25}$, $\mathbf{x}\in\mathbb{R}^{25}$ ($\mathbf{x}=\operatorname{vec}\left(\mbX\right)$ where $\mbX\in\mathbb{R}^{5\times 5}$), and $\mathbf{y}\in\mathbb{R}^{200}$. We consider the number of time-varying sampling threshold sequences to be $m\in\left\{10,20,30,40,50,60\right\}$. Each row of $\mbA$ is generated as $\mathbf{a}_{j}\sim\mathcal{N}\left(\mathbf{0},\mbI_{25}\right)$. For the desired matrix $\mbX$, we generate $\mbX=\mbK\mbK^{\top}$ where (i) $\mbK\in\mathbb{R}^{5\times 4}$ is the Gaussian matrix, and (ii) $\mbK\in\mathbb{R}^{5\times 1}$ is the Gaussian vector. Also, each time-varying sampling threshold $\boldsymbol{\uptau}^{(\ell)}$ is considered to have the distribution $\boldsymbol{\uptau}^{(\ell)}\sim\mathcal{N}\left(\mathbf{0},\mbI_{200}\right)$. Fig.~\ref{figure_2} appears to confirm the possibility of recovering the desired matrix $\mbX^{\star}$ in the one-bit polyhedron (\ref{eq:80n}) by ORKA. As expected, the performance of the recovery will be significantly enhanced as the number of time-varying sampling threshold sequences grows large. Also, similar to before, it can be seen that the Block SKM outperforms the PrSKM in the low-rank matrix recovery problem.

To improve the recovery performance, we proposed the adaptive time-varying sampling threshold in Section~\ref{ada_thresh}. Fig.~\ref{figure_3} illustrates the performance of the Block SKM in the low rank matrix recovery in the one-bit polyhedron (\ref{eq:80n}) when we have the high-dimensional input signal $\mathbf{x}\in \mathbb{R}^{128}$ and $\mbA\in\mathbb{R}^{20000\times 128}$, with (i) a random threshold, and (ii) an adaptive time-varying threshold. As can be seen, the recovery performance is significantly enhanced when the Block SKM is equipped with the adaptive time-varying threshold.

\begin{figure*}[t]
	\centering
	\begin{subfigure}[b]{0.45\textwidth}
		\includegraphics[width=1\linewidth]{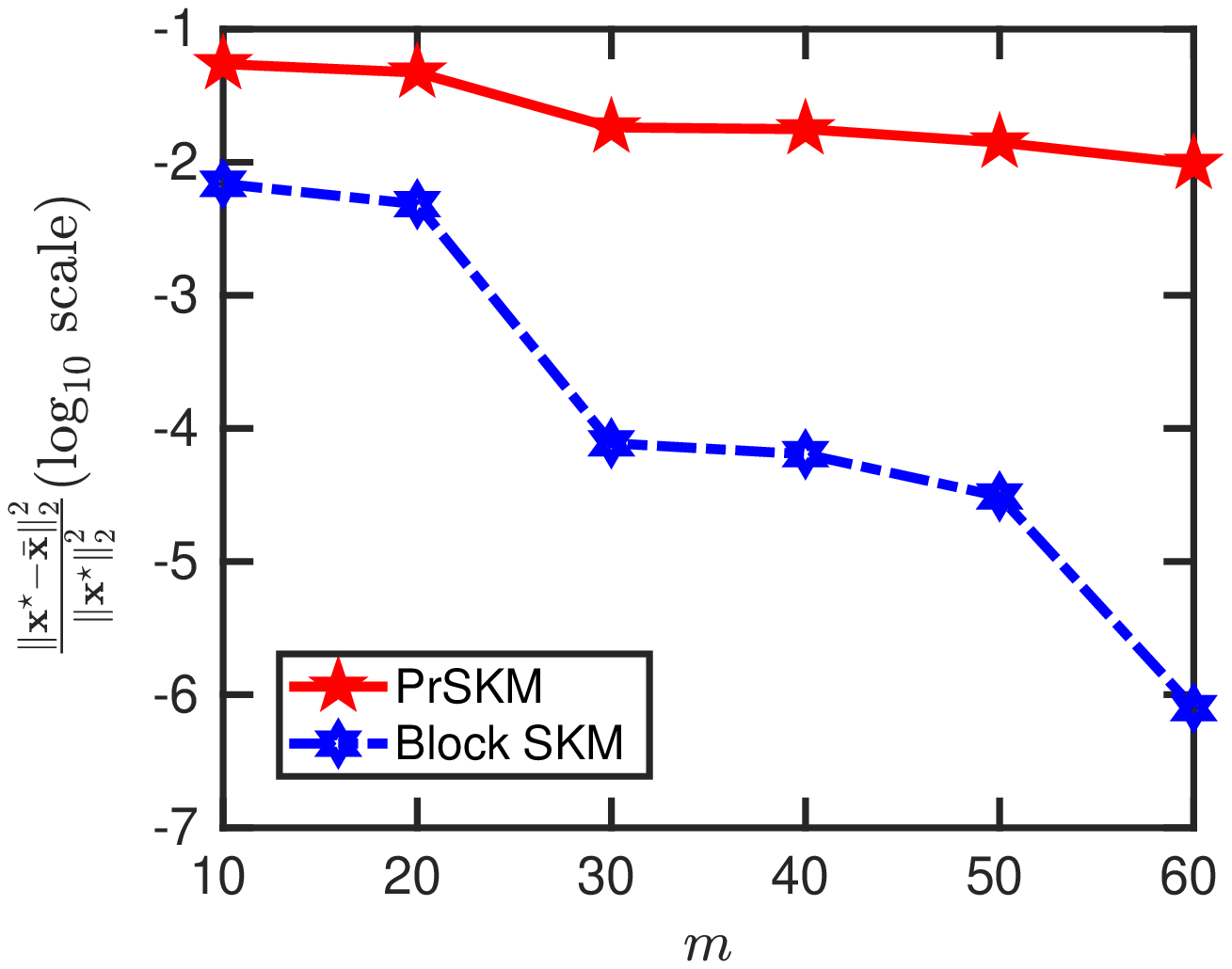}
		\caption{}
	\end{subfigure}
	\begin{subfigure}[b]{0.45\textwidth}
		\includegraphics[width=1\linewidth]{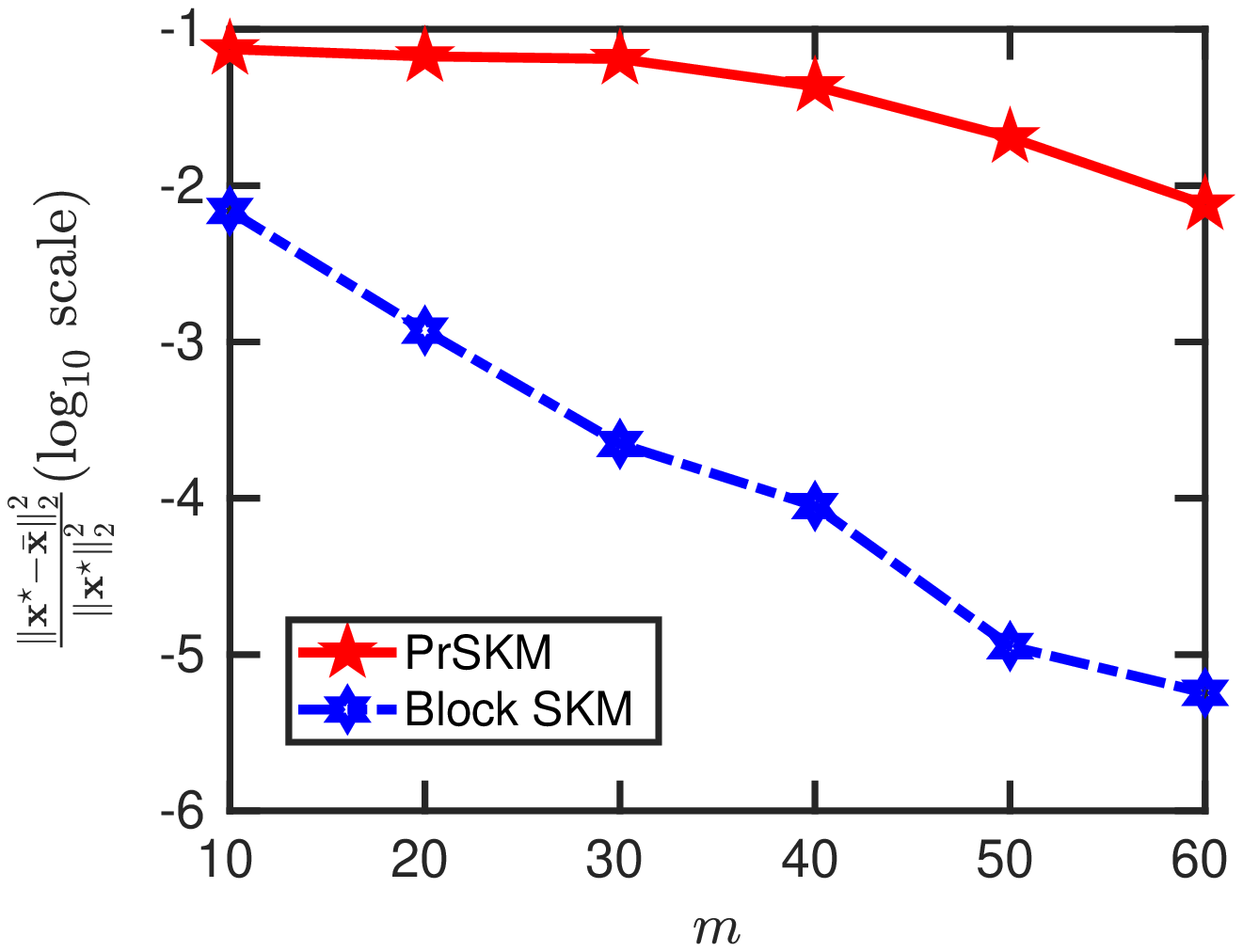}
		\caption{}
	\end{subfigure}
	\caption{Average NMSE for the error between the desired signal $\mathbf{x}^{\star}$ and its recovered version $\bar{\mathbf{x}}$ for different time-varying sampling threshold sequences sizes when the PrSKM and the block SKM are utilized in ORKA with (a) $k=2$, (b) $k=4$.}
	\label{figure_4}
\end{figure*}

\section{One-Bit Compressed Sensing:\\ 
From Optimization to Linear Feasibility}

Compressed sensing (CS) is an interesting and rapidly growing area of research that has attracted considerable attention in electrical engineering, applied mathematics, statistics, and computer science \cite{eldar2012compressed,davenport2016overview}. In CS, a sparse high-dimensional signal is to be recovered by incomplete measurements such a recovery may be formulated as \cite{eldar2012compressed}:
\begin{equation}
\label{eq:1nnnnnnnm}
\begin{aligned}
\min_{\mathbf{x}}\quad & \left\|\mathbf{x}\right\|_{1}\\
\text{s.t.} \quad &\mbA\mathbf{x}=\mathbf{y},
\end{aligned}
\end{equation}
where $\mbA\in\mathbb{R}^{m\times n}$, and $m\ll n$. One of the important applications of CS emerges in the signal recovery from a sequence of acquisitions $\{y_{i}\}$ obtained from a sparse linear transformation (wavelet transformations are known for such a property, for instance) in the magnetic resonance imaging (MRI). The reconstruction problem of the desired signal $\mathbf{x}^{\star}$ is given by
\begin{equation}
\label{eq:1nnnnnnnmm}
\begin{aligned}
\min_{\mathbf{x}}\quad & \left\|\mathbf{x}\right\|_{1}\\
\text{s.t.} \quad &\mathcal{A}_{i}\left(\mathbf{x}\right)=y_{i},\quad i\in\left\{1,\cdots,n\right\}.
\end{aligned}
\end{equation}

In this section, we first formulate the optimization problem of the \emph{one-bit compressed sensing}. Then, by taking advantage of one-bit sampling, we increase sample size in (\ref{eq:1nnnnnnnmm}) and create an  associated one-bit polyhedron.
\subsection{Problem Formulation}
Let $\boldsymbol{\uptau}$ denotes the time-varying threshold vector. The one-bit samples are generated as
\begin{equation}
\label{eq:101}
\begin{aligned}
r_{i} &= \begin{cases} +1 & \mba^{\top}_{i}\mathbf{x}\geq\uptau_{i}, \\ -1 & \mba^{\top}_{i}\mathbf{x}<\uptau_{i},
\end{cases}
\end{aligned}
\end{equation}
where $\mathcal{A}_{i}\left(\mathbf{x}\right)=\mba^{\top}_{i}\mathbf{x}$.
The occurrence probability vector $\mbp$ for the one-bit measurement $\mbr$ is given as \cite{9896984},
\begin{equation}
\label{eq:102}
\begin{aligned}
p_{i} &= \begin{cases} \Phi\left(\mba^{\top}_{i}\mathbf{x}\right) & \text{for}\quad \{r_{i}=+1\}, \\ 1-\Phi\left(\mba^{\top}_{i}\mathbf{x}\right) &  \text{for}\quad \{r_{i}=-1\},
\end{cases}
\end{aligned}
\end{equation}
where $\Phi(.)$ is the CDF of $\boldsymbol{\uptau}$. The log-likelihood function of the sign data $\mbr$ is given by
\begin{equation}
\label{eq:103}
\begin{aligned}
\mathcal{L}_{\mbr}(\bmu,\mathbf{x}) &= \sum^{m}_{i=1}\left\{\mathbb{I}_{(r_{i}=+1)}\log\left(\Phi(\mba^{\top}_{i}\mathbf{x})\right) \right.\\& \left.+\mathbb{I}_{(r_{i}=-1)}\log\left(1-\Phi(\mba^{\top}_{i}\mathbf{x})\right)\right\}.
\end{aligned}
\end{equation}
Therefore, the maximum likelihood estimation (MLE) for the one-bit compressed sensing can be written as
\begin{equation}
\label{one-cs}
\min_{\mathbf{x}} \mathcal{L}_{\mbr}(\bmu,\mathbf{x})+\lambda \|\mathbf{x}\|_{1}.
\end{equation}
The alternative formulations for one-bit compressed sensing can be found in \cite{khobahi2020model}.

Nevertheless, as discussed earlier, by deploying one-bit sampling, the opportunity exists to increase the number of samples in (\ref{eq:1nnnnnnnmm}). The one-bit compressed sensing is thus solely accomplished by creating a highly-constrained one-bit polyhedron. In other words, instead of solving an optimization problem with costly constraints, the problem may be tackled by the proposed accelerated Kaczmarz algorithms; namely, PrSKM and the block SKM.
\begin{figure*}[t]
	\centering
	\begin{subfigure}[b]{0.45\textwidth}
		\includegraphics[width=1\linewidth]{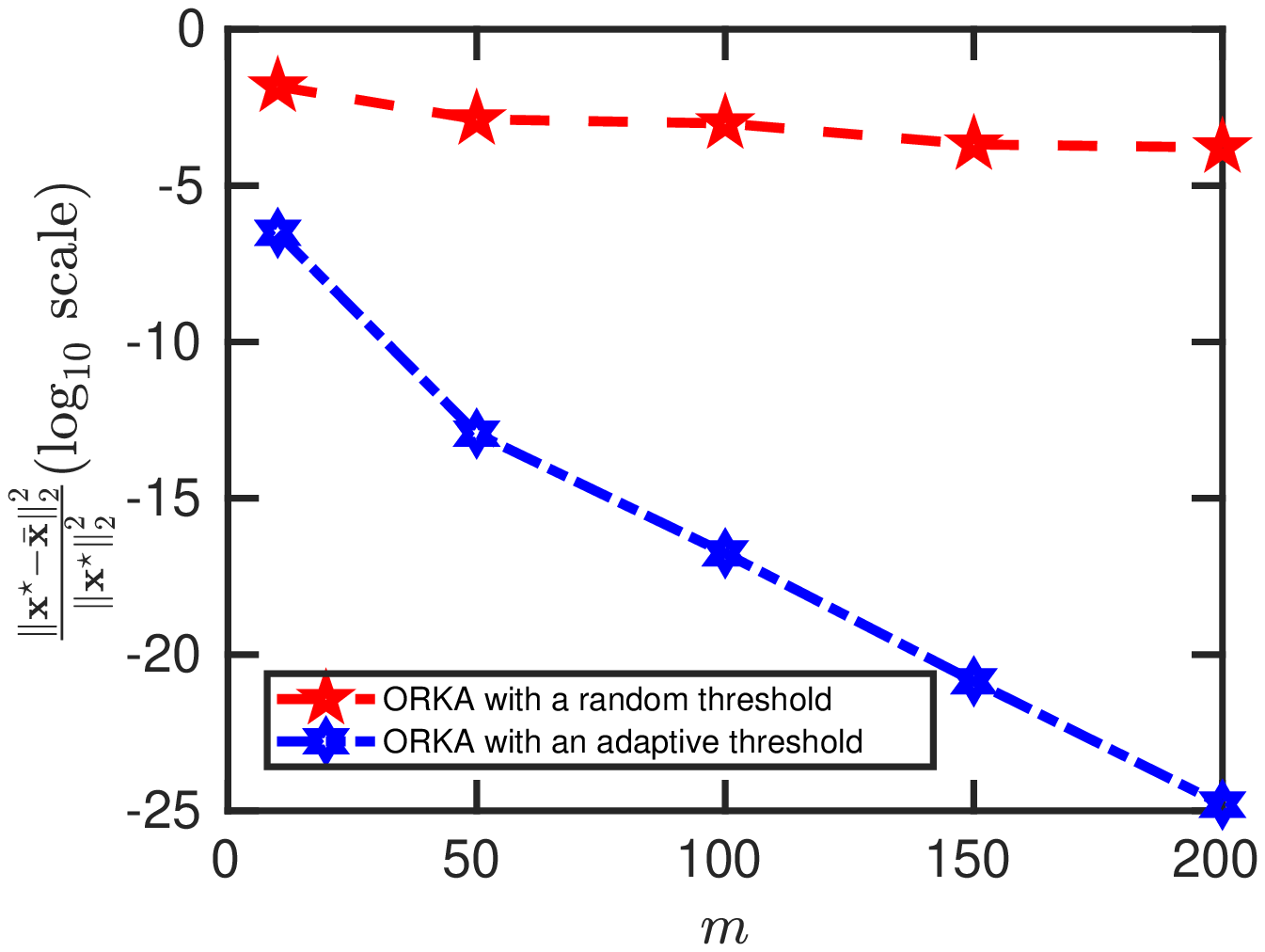}
		\caption{}
	\end{subfigure}
	\begin{subfigure}[b]{0.45\textwidth}
		\includegraphics[width=1\linewidth]{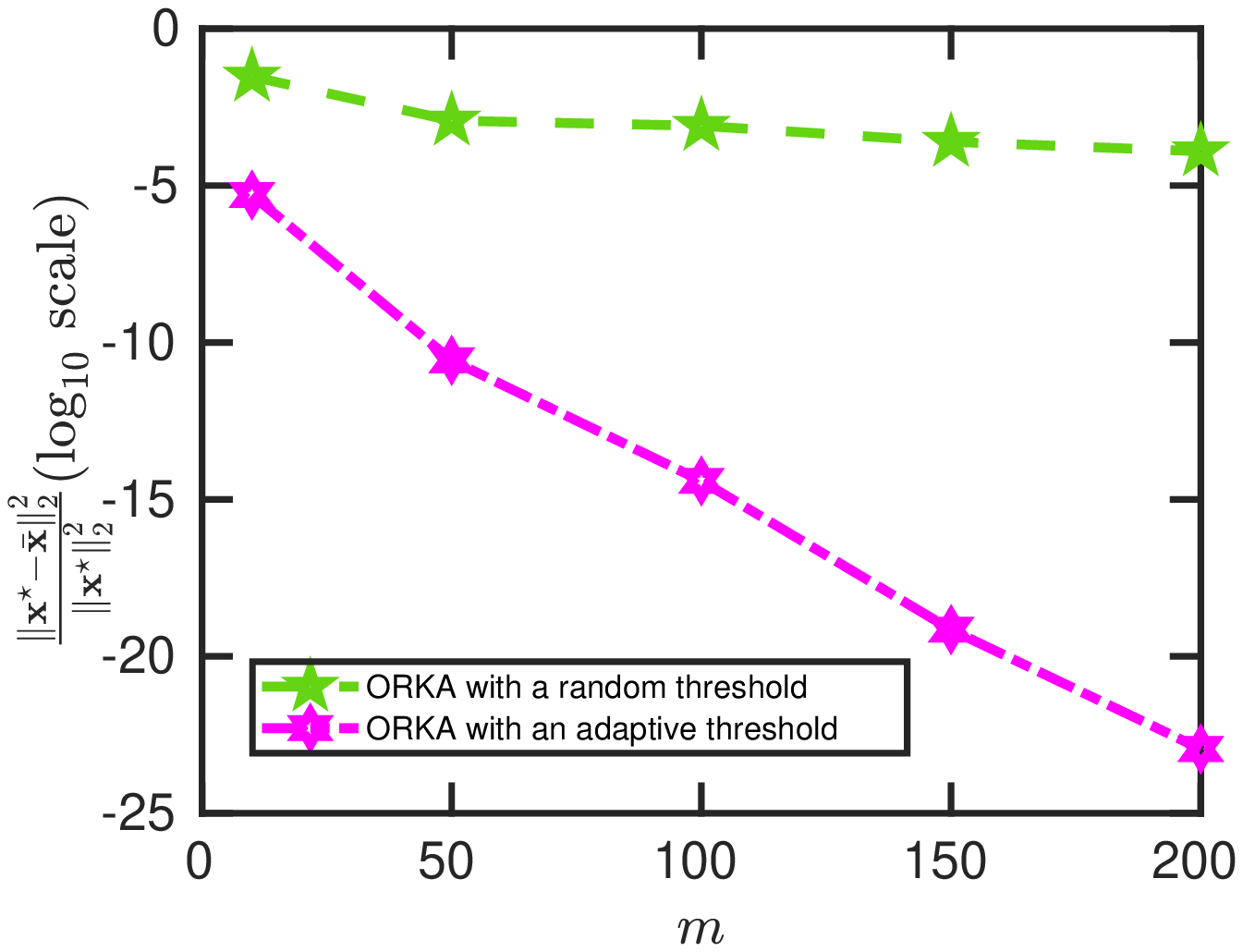}
		\caption{}
	\end{subfigure}
	\caption{Comparing the average NMSE for the desired signal $\mathbf{x}^{\star}$ and its recovered signal using ORKA when (i) a random threshold and (ii) the adaptive sampling threshold are adopted with (a) $k=20$, (b) $k=40$}
	\label{figure_5}
\end{figure*}
\subsection{Numerical results}
To examine the performance of ORKA in CS and to validate the theoretical results described in this paper, we consider signal recovery with different number of time-varying sampling threshold sequences $m\in\left\{10,20,30,40,50,60\right\}$. Input signals $\mathbf{x}^{\star}\in\mathbb{R}^{10}$ are generated with sparsity orders $k=2$ and $k=4$, respectively. The sparsity order $k$ is defined as the number of non-zero elements in a vector. Time-varying sampling thresholds and the constraint matrix $\mbA$ are generated as in Subsection~\ref{NUM_matrix}. To compare two proposed algorithms, the NMSE defined in (\ref{eq:1700000}) is utilized and the results are averaged over $15$ experiments.

As can be seen in Fig.~\ref{figure_4}, by increasing the number of time-varying sampling threshold sequences, the performance of our method is improved. Beside the possibility of increasing the number of measurements $n$, the higher number of samples are available in ORKA by comparing the measurements with multiple threshold sequences $\ell \in\left\{1,\cdots,m\right\}$. In other words, we have the opportunity to increase $n$ and $m$ simultaneously, when the number of samples is $m^{\prime}=m n$.

Same as Subsection~\ref{NUM_matrix}, the adaptive thresholding algorithm is applied to ORKA for the high-dimensional input signal $\mathbf{x}\in\mathbb{R}^{128}$ in order to enhance its recovery performance, whose outcome is presented in Fig.~\ref{figure_5}. The NMSE results are reported with sparsity orders $k=20$ and $k=40$.

\begin{table} [t]
\caption{Comparing CPU times and $\operatorname{NMSE}$ of ORKA and $\ell^{1}$-minimization.}
\centering
\begin{tabular}{ | c | c | c | c |}
\hline
\text {Algorithm} & \text {$m^{\star}$} & \text {CPU time (s)} & \text {$\operatorname{NMSE}$} \\[0.5 ex]
\hline \hline
\text{ORKA} & $500$ &  $3.1240e-04$ & $3.2052e-12$ \\[1 ex]
\hline
\text{$\ell^{1}$} & $100$ &  $0.0071$ & $2.4572e-11$ \\[1 ex]
\hline
\end{tabular}
\label{table_1}
\end{table}
To further investigate the efficacy of ORKA in CS, we compare our proposed approach with the well-known $\ell^{1}$-minimization approach formulated in (\ref{eq:1nnnnnnnmm}) in terms of NMSE and CPU time. As presented in Table~\ref{table_1}, ORKA outperforms $\ell^{1}$-minimization in terms of both NMSE and CPU time. The results are obtained for $\mathbf{x}\in\mathbb{R}^{128}$ when the optimal number of samples are utilized, and where $m^{\star}=4k\log(n/k)$ and $m^{\star}=500$ are considered for the high-resolution method and ORKA, respectively. Herein, optimality of sample sizes means that the number of samples utilized by algorithms leads to their best performance, i.e. satisfying the criterion $\left\|\mathbf{x}_{i}-\mathbf{x}^{\star}\right\|_{2}^{2}\leq 5\times 10^{-11}\left\|\mathbf{x}^{\star}\right\|_{2}^{2}$. By this comparison, we remove the burden of the large number of samples from the $\ell^{1}$-minimization to fairly compare their optimal shape deploying incomplete measurements with that of ORKA.

It is worth pointing out that for a $64$-bit ADC, $m=100$ corresponds to $6400$ bits of information while ORKA solely employs $500$ bits. Therefore, it appears from Table~\ref{table_1} that ORKA achieves a better accuracy in terms of NMSE with not only fewer information bits but also a smaller computational cost.

\section{Conclusion}
We proposed a novel algorithm, ORKA, that takes advantage of the abundant number of samples available in one-bit sampling with time-varying thresholds to efficiently and globally solve some well-studied problems in the form of (\ref{eq:1nnnn}); including low-rank matrix recovery and compressed sensing. Moreover, two state-of-the-art randomized Kaczmarz algorithms are proposed to use in ORKA to find the desired signal inside the emerging confined feasible regions, named the one-bit polyhedron, with an enhanced convergence rate. The numerical results showcased the effectiveness of the proposed approaches for the low-rank matrix recovery and compressed sensing problems.

\bibliographystyle{IEEEbib}
\bibliography{strings,refs}

\end{document}